\documentclass[a4paper]{sn-jnl}
\usepackage{packages}
\usepackage{macro}

\hypersetup{
  colorlinks=true,
  linkcolor=blue,
  citecolor=red
}
\theoremstyle{thmstyleone}%
\newtheorem{theorem}{Theorem}[section]
\newtheorem{proposition}[theorem]{Proposition}%
\newtheorem{lemma}[theorem]{Lemma}%
\newtheorem{corollary}[theorem]{Corollary}%

\theoremstyle{thmstyletwo}%
\newtheorem{example}[theorem]{Example}%
\newtheorem{remark}[theorem]{Remark}%

\theoremstyle{thmstylethree}%
\newtheorem{definition}[theorem]{Definition}%
\newtheorem{notation}[theorem]{Notation}%

\begin{document}

\title{New perspectives for code locality in the rank metric}

\author*[1]{\fnm{Camille} \sur{Garnier}}\email{camille.garnier@unilim.fr}

\author*[2]{\fnm{Julien} \sur{Lavauzelle}}\email{julien.lavauzelle@univ-paris8.fr}

\author*[3]{\fnm{Jade} \sur{Nardi}}\email{jade.nardi@univ-rennes.fr}

\author*[1]{\fnm{Ilaria} \sur{Zappatore}}\email{ilaria.zappatore@unilim.fr}

\affil[1]{\orgdiv{XLIM, CNRS UMR 7252}, \orgname{Université de Limoges}, \orgaddress{\street{123, avenue Albert Thomas}, \city{Limoges}, \postcode{87060}, \country{France}}}

\affil[2]{\orgdiv{LAGA, CNRS UMR 7539}, \orgname{Université Paris 8}, \orgaddress{\street{2 Rue de la Liberté}, \city{Saint-Denis}, \postcode{93200}, \country{France}}}

\affil[3]{\orgdiv{CNRS, IRMAR UMR 6625}, \orgname{Univ Rennes}, \postcode{F-35000} \city{Rennes}, \country{France}}

\abstract{
  In coding theory, local recovery enables the efficient recovery of some part of (lost) coded data by accessing only a small number of other data entries. Locality was mostly but intensively studied for the recovery of individual symbols, that is, in the context of the Hamming metric.

  In this work, we propose a new definition of locality for general rank-metric codes. This definition differs from a previous work of Kadhe, El Rouayheb, Duursma and Sprintson [IEEE Trans. Inf. Theory 2019], by allowing to efficiently recover any element of the support, and without relying on any choice of bases of the underlying vector spaces.

  Our work firstly relies on a precise study of code puncturing and shortening for codes viewed as spaces of linear maps. We then provide examples and general constructions, showing the difference between our notion and that of Kadhe \emph{et~al.} We then derive a Singleton-like bound for rank locally recoverable codes, and we finally prove that a construction similar to classical Tamo--Barg codes is optimal with respect to this bound.
}


\maketitle

\section{Introduction}

Locally recoverable codes (LRCs) are error-correcting codes in which each codeword symbol can be recovered by accessing only a small subset of the remaining symbols.
The property of local recoverability makes these codes well suited to distributed storage systems (DSSs). When a storage node becomes unavailable, it can be repaired by accessing only a limited number of other nodes, reducing the repair cost.

The ideas underlying LRCs originate from~\cite{Reliable_Memories_with_Subline_Accesses} and~\cite{Pyramid_Codes}, which introduced the principle of reducing the amount of information accessed during data recovery. LRCs were first systematically studied by Gopalan \textit{et al.}~\cite{locality12}. They established a bound relating the code parameters (length, dimension and minimum distance) to its \emph{locality}, i.e., the number of other coordinates that must be accessed to recover a given one. Families of codes achieving this bound are given in~\cite{Tamo_2014, Pyramid_Codes}, and further constructions appear in~\cite{LRC_surfaces, Barg_2015, lrc_availability, lrc_sep, lrc_curves_surfaces, Tamo_2016, rationallrc}. Note that the term \enquote{locally repairable codes} is also commonly used in the literature. In some works, it is used as a synonym for locally recoverable codes (e.g.,~\cite{rationallrc,matroid,repairableLRC}), whereas in others it is considered as a more general notion (e.g.,~\cite{codes_for_DSS,Prakash12}).

All the aforementioned works study local recoverability in the Hamming metric. However, the Hamming metric does not adequately model  all types of errors encountered in communication and storage systems. For instance, in network coding, space-time coding, and cryptography, errors may affect entire subspaces rather than individual coordinates. Such errors can simultaneously corrupt many coordinates, and a repair scheme based on Hamming-style locality may require access to more information than necessary to characterize the error. The rank metric provides a more suitable model in these scenarios; see~\cite{bartz2022rank} for a survey of its applications.

Rank-metric codes are typically considered in two equivalent representations. A matrix rank-metric code, or matrix-code, consists of codewords that are matrices over a finite field, with the rank-metric defined as the rank of the difference between two codewords. A vector rank-metric code, or vector-code, is defined over an extension field and endowed with the rank-metric over the base field. After fixing a basis of the extension field over the base field, each coordinate of a codeword can be expanded with respect to this basis. Consequently, every codeword is represented by a matrix whose columns correspond to the coordinates of the original vector.
Several works have studied locality in connection with the rank metric. In~\cite{lrc_via_rank_metric} and~\cite{opti_lrc_via_rank_metric}, rank-metric codes are used as building blocks for the construction of codes with locality properties for DSSs. However, the resulting notion of local recoverability used in these works is a straightforward translation of the classical Hamming locality. A different approach is taken in~\cite{Kadhe_et_al19}, where locality is formulated for matrix codes. In this framework, a \emph{column} is locally recoverable if it can be reconstructed by accessing only a small number of other columns. With this point of view,~\cite{Kadhe_et_al19} provides a Singleton-like upper bound on the minimum rank distance of rank-LRCs and an optimal construction. More recently, Bastioni~\emph{et al.}~\cite{loc2} proposed new optimal constructions using Drinfeld modules.


In this work we introduce a new notion of local recoverability, that is more intrinsic to the rank metric. To establish the analogy, the classical notion of LRCs in the Hamming metric is first recalled. Consider a code $\C\subseteq \Fq^n$ in the Hamming metric.
A coordinate $i\in \{1, \dots, n\}$ has \emph{locality $r\geq 1$ in $\C$} if there exists a subset $\mathcal{R}_i \subseteq \{1, \dots, n\}\backslash \{i\} $ of size at most $r$ (called helper set) such that the punctured codes $\C_{\mid \mathcal{R}_i}$ and $\C_{\mid \mathcal{R}_i\cup \{i\}}$ have the same dimension. This implies that the $i^{th}$ coordinate of any codeword can be recovered via linear operations from the values of the $r$ other coordinates in $\mathcal{R}_i$. A code $\C\subseteq \Fq^n$ is said to be  \emph{$r$-locally recoverable} if each coordinate $i\in \{1, \dots, n\}$ has locality $r$ in $\C$.

In the Hamming metric, the support of a codeword is its set of nonzero coordinates, and its weight is the size of this support. For vector codes, the (column) support of a codeword is the subspace over the base field spanned by its coordinates, and its rank weight is the dimension of this subspace. In the corresponding matrix representation of the code, this support can be identified with the column space of the associated matrix. 
Since local recoverability is closely related to the notion of support, we propose a new formulation in rank metric that reflects the different nature of the rank support compared to the Hamming one. We therefore define locality at the level of subspaces rather than coordinates/columns.

To avoid dependence on the choice of bases, codewords are regarded as linear maps $f : U \rightarrow V$ between finite-dimensional vector spaces $U$ and $V$. In this framework, a rank-metric code is a subspace of $\Hom(U,V)$. Roughly speaking, a vector $u \in U$ is said to have \emph{locality} $r$ in a code $\C \subseteq \Hom(U,V)$ if, for every $f \in \C$, the value $f(u)$ can be recovered from $f(S)$ alone, where $S$ is an $r$-dimensional subspace of $U$ that does not contain $u$. The code $\C$ is \emph{$r$-rank locally recoverable} if every vector $u \in U$ has locality $r$ in $\C$ (see Definition~\ref{def lrc rank} for more details). This definition matches the nature of supports in the rank metric: a code is rank-locally recoverable with locality $r$ if every element of a codeword's support can be recovered from a \enquote{piece} of that support of size $r$. It also fits naturally the notion of rank puncturing introduced in~\cite{coveringradius}, which is used to formally define local recoverability as in the Hamming case.

\paragraph{Our contributions} The contribution of this paper is threefold.

First, we develop a general framework for rank-metric codes whose codewords are linear maps between finite-dimensional vector spaces. In this framework, we reformulate the notions of puncturing and shortening previously introduced for matrix codes in~\cite{coveringradius} and~\cite{Neri_tensor} and for vector codes in~\cite{Neri19}. This formulation avoids any dependence on the choice of bases and allows us to treat matrix codes and vector codes in a unified way.

Second, we introduce a new notion of local recoverability for rank-metric codes. Rather than defining locality on coordinates, as in the Hamming metric, we define it for every nonzero vector in the domain of the code. As explained above, this definition is coordinate-free, and it is consistent  with both the subspace nature of  rank-support and the notion of rank puncturing. We characterize locality in terms of the dual code and extend the classical notion of an information set to the rank-metric setting through the notion of \emph{information space}.

Third, we prove a Singleton-like bound for $\Fq$-linear rank-metric codes with locality. We provide constructions of rank locally recoverable codes achieving the bound and show that our construction recovers~\cite[Construction~1]{bartz2022rank} as a particular case.


\section{Preliminaries}

Let $q$ be a prime power, $\Fq$ be a finite field of order $q$ and $\Fqm$ be the extension field of $\Fq$ of degree $m$. In this paper, vectors over the extension field $\Fqm$ are represented by lowercase bold letters: $\xv, \yv$. Matrices are represented by uppercase letters $M, G, H$. We also denote the space of $m \times n$ matrices over a field $\F$ as $\F^{m \times n}$. For a positive integer $n$, the set $\set{1,\dots,n}$ is denoted by $[n]$.

Let $U$ and $V$ be $\Fq$-vector spaces of dimension $n$ and $m$. Throughout this paper, we fix bases $\B_U=(u_1,\dots,u_n)$ and $\B_V=(v_1,\dots,v_m)$ of $U$ and $V$, respectively. We endow $U$ and $V$ with the unique (nondegenerate) symmetric bilinear forms $\beta_U$ and $\beta_V$ for which $\B_U$ and $\B_V$ are orthonormal, i.e.,
\[
\beta_U(u_i,u_j)=\delta_{i,j} \text{ and } \beta_V(v_\ell,v_k)=\delta_{\ell,k}, \quad i,j \in [n],\; k,\ell \in [m]. 
\]
These are the \emph{standard inner products} on $U$ and $V$ with respect to the chosen bases. They are \emph{nondegenerate}, i.e., any nonzero $u \in U$ (resp. $v \in V$) gives rise to a nonzero map $\beta_U(u,-)$ (resp. $\beta_V(v,-))$. However they may (and here will) be nondefinite, i.e., there may exist a nonzero $u \in U$ (resp. $v \in  V$) such that $\beta_U(u,u)=0$ (resp. $\beta_V(v,v)=0$).

We denote by $\Hom(U,V)$ the space of $\Fq$-linear maps from $U$ to $V$. If $U = V$, the associated algebra is denoted by $\End(U)$. We also denote by $\GL(U)$ and $\GL_n(\Fq)$ the group of inversible endomorphisms over $U$ and over $\Fq^n$, respectively.

\subsection{Adjoint of a homomorphism}\label{subsec : adjoint}

In this section, we recall some classical notions of linear algebra that we will use throughout this paper. Most of the results of this section are taken from~\cite{roman_advanced_2007}.

\begin{definition}[Adjoint of a linear map]
  Let $f\in \Hom(U,V)$. The \emph{adjoint} of $f$ is the map $f^\dagger\in \Hom(V,U)$ such that
  \[
  \beta_V(f(u),v)= \beta_U(u,f^\dagger(v)), \qquad \forall u\in U, \: \forall v \in V. 
  \]
\end{definition}

The image and the kernel of a linear map and its adjoint are related via the orthogonal complement operation.

\begin{definition}[Orthogonal complement]
  Let $W \subseteq U$. The orthogonal complement of $W$ in $U$ is
  \[
  W^{\perp}\eqdef \{u \in U \mid \beta_U(u,w)=0, \forall w \in W\}.
  \]
\end{definition}

The following results follow from the correspondence between the dual and the adjoint maps (see~\cite[\textsection 10]{roman_advanced_2007}) and properties of the dual map (namely~\cite[Theorems~3.20 and 3.22]{roman_advanced_2007}).

\begin{theorem}\label{thm:properties_adjoint}
  Let $f \in \Hom(U,V)$. Then,
  \begin{itemize}
  \item $\ker(f^\dagger) = \im(f)^\perp$,
  \item $\im(f^\dagger) = \ker(f)^\perp$,
  \item $\rk(f^\dagger)= \rk(f)$,
  \item $f$ is surjective if and only if $f^\dagger$ is injective (and vice versa).
  \end{itemize}
\end{theorem}

The following lemma relates the matrix representations of a linear map and its adjoint. 

\begin{lemma}\label{lem:adjoint_matrices}
  Let $\B_U, \B_V$ be two orthonormal bases of $U$ and $V$. For every $f \in \Hom(U,V)$, we have
  \[
  \Mat{\B_V}{ \B_U} {f^\dagger} =  \Mat{\B_U}{ \B_V} {f}^\top.
  \]
\end{lemma}

\subsection{Rank-metric codes}

In this section, we recall basic definitions and results on rank-metric codes. We first introduce rank-metric codes as spaces of $\Hom(U,V)$. This viewpoint provides a unified framework, from which the classical matrix and vector representations of rank-metric codes, commonly used in rank-metric literature, can be recovered. We then introduce support spaces related to a linear map, defined in terms of the image and the kernel of the adjoint of a linear map. This will allow us to recover the classical notions of row and column support in the matrix setting.
\subsubsection{Representations of rank-metric codes}
In this section, we briefly recall the basic notions of rank metric, rank-metric codes in their different representations, as well as code equivalence and duality.

\paragraph*{Rank-metric codes with homomorphisms}
Given $f,g \in \Hom(U, V)$,  the \textit{rank distance} between $f$ and $g$ is
$$
\dr(f,g) \eqdef \rk(f-g).
$$
This function defines a metric on $\Hom(U, V)$. The \emph{rank weight} of $f\in \Hom(U, V)$ is defined as
\[
\wr(f) \eqdef \dr(f,0)=\rk(f).
\]

\begin{definition}[Rank-metric codes]
  A \emph{rank-metric code} $\C$ is an $\Fq$-linear subspace of $\Hom(U, V)$ endowed with the rank metric. The \emph{dimension} of this code is its dimension as an $\Fq$-vector space. Its \emph{minimum rank distance} (shortly \emph{minimum distance}) is
  \[
  \dr(\C) \eqdef \min_{f\in \C\setminus\{0\}} \wr(f).
  \]
\end{definition}

The notion of \emph{equivalence of rank-metric codes} is defined in terms of $\Fq$-linear isometry, which we now recall.

\begin{definition} [$\Fq$-linear isometry]
  An \emph{$\Fq$-linear isometry} $\varphi$ between $\Hom(U,V)$ and $\Hom(U',V')$ is an $\Fq$-linear homomorphism $\varphi:\Hom(U,V)\to \Hom(U',V') $ such that $\wr(\varphi(f)) = \wr(f)$ for every $f\in \Hom(U,V)$.
\end{definition} 

\begin{definition}[Equivalent rank-metric codes]
  Two rank-metric codes $\C \subseteq \Hom(U, V)$ and $\mathcal{D} \subseteq \Hom(U',V')$  are \emph{equivalent} if there exists an $\F_q$-linear isometry $\varphi$ such that $\varphi(\C) =\mathcal{D}$.
\end{definition}

\begin{remark}
  \label{rk:isometries}%
  We recall that every $\F_q$-linear isometry is induced by compositions with invertible linear maps acting on the domain and the codomain~\cite[Proposition~III.15]{morrison2014equivalence}. In particular, if $\dim_{\Fq}(U) \neq \dim_{\Fq}(V)$, then every $\F_q$-linear isometry $\varphi$ is of the form
  \begin{equation}
    \label{eq:left-right-equiv}%
    \varphi(f) = \beta \circ f \circ \alpha,
  \end{equation}
  for some invertible maps $\alpha \in \Hom(U', U)$ and $\beta \in \Hom(V, V')$. Note that if $\dim_{\Fq}(U) = \dim_{\Fq}(V)$, then isometries of the form
  \[
  \varphi(f) = \alpha \circ f^{\dagger} \circ \beta
  \]
  must also be considered.
\end{remark}

If two rank-metric codes $\C \subseteq \Hom(U, V)$ and $\mathcal{D} \subseteq \Hom(U',V')$ are such that $\varphi(\C) = \mathcal{D}$ for $\F_q$-linear isometries $\varphi$ of the form~\eqref{eq:left-right-equiv}, we say they are \emph{left-right-equivalent}.

\paragraph{Duality} We now observe that the bilinear forms on $U$ and $V$ defined in the previous section allow us to define the following \emph{nondegenerate symmetric bilinear form} on $\Hom(U,V)$,
\[
\beta_{\Hom(U,V)}(f,g) \eqdef \Tr(f \circ g^\dagger), \qquad \forall f,g \in \Hom(U,V),
\]
where $\Tr$ denotes the usual trace map of the endomorphism $f \circ g^\dagger \in \End(V)$. We can now define the notion of duality of rank-metric codes.
\begin{definition}[Dual code]
  Let $\C\subseteq \Hom(U,V)$ be a rank-metric code. The \emph{dual code} of $\C$ is defined as
  \[
  \C^\perp = \{ h \in \Hom(U, V) \mid  {\rm Tr}(h \circ f^\dagger) = 0, \forall f \in \C \}.
  \]
\end{definition}

\paragraph*{Matrix rank-metric codes}
Fix a basis $\B_U=(u_1, \ldots, u_n)$ of $U$ and a basis $\B_V=(v_1, \ldots, v_m)$ of $V$, then every map $f\in\Hom(U,V)$ can be represented by a matrix $\Mat{\B_U}{\B_V}{f}\in \Fqmn$. Indeed for each basis vector $u_j$, we can write
\[
f(u_j)=\sum_{i=1}^m a_{i,j}v_i,
\]
and the matrix associated to $f$ is $\Mat{\B_U}{\B_V}{ f}=(a_{i,j})\in \Fqmn$. This correspondence defines an $\Fq$-isomorphism between $\Hom(U, V)$ and $\Fqmn$. Under this isomorphism, the rank of $f\in\Hom(U,V)$ coincides with the rank of the matrix $\Mat{\B_U}{\B_V}{f}$, independently of the choice of bases $\B_U$ and $\B_V$. This allows us to identify $\Hom(U,V)$ with $\Fqmn$ equipped with the rank metric. Rank-metric codes then correspond to \emph{matrix rank-metric codes} (shortly \emph{matrix codes}), \ie linear subspaces of $\Fqmn$ endowed with the \emph{rank metric}, as introduced by Delsarte~\cite{Delsarte}.



\begin{remark}
  Let $f,g \in \Hom(U,V)$ and let $F=\Mat{\B_U}{\B_V}{f}$ and $G=\Mat{\B_U}{\B_V}{g}$ the corresponding matrices. Since the bases $\B_U$ and $\B_V$ are orthonormal, Lemma~\ref{lem:adjoint_matrices} implies
  \begin{equation}\label{eq: trace_prod}
    \beta_{\Hom(U,V)}(f,g)=\Tr(FG^\top).
  \end{equation}
\end{remark}
Therefore, the dual code previously defined corresponds to the classical Delsarte dual code for rank-metric codes, \ie
\begin{equation}\label{eq:dual_matrix}
  \C^\perp = \{H \in \Fq^{m\times n} \mid \Tr(HF^\top)=0, \forall F\in \C\}.
\end{equation}

\paragraph*{Vector rank-metric codes} The set $\Fqm^n$ of vectors with entries in the extension field $\Fqm$ can also be endowed with the rank metric. Given a basis $\Gamma = (\gamma_1, \dots, \gamma_m)$ of $\Fqm$ over $\Fq$ and a vector ${\bm x} = (x_1, \dots, x_n) \in \F_{q^m}^n$, we denote by $\Gamma(\bm x)\in \Fqmn$ the matrix whose $i$-th column is the coordinate vector of $x_i$ with respect to $\Gamma$. In this setting, the \emph{rank weight} of a vector $\cv=(c_1, \ldots, c_n)\in \Fqm^n$ is defined equivalently as
\[
\wr(\cv)\eqdef \dim_{\Fq}( \vect{c_1, \ldots, c_n}{\Fq}) = \rk(\Gamma(\cv)).
\]
Therefore, every $\Fq$-linear subspace $\C$ of $\Fqm^n$ corresponds, via the basis $\Gamma$, to a matrix rank-metric code in $\Fqmn$, that we denote by $\Gamma(\C)$. We can make these codes fit within the language of morphisms by associating to each codeword $\bm c=(c_1,\dots,c_n) \in \Fqm^n$ a morphism $f_{\bm c} \in \Hom(U,\Fqm)$ defined by $f_{\bm c}(u_i)=c_i$.

In the literature, there is a particular interest on subspaces of $\Fqm^n$ that are not only $\Fq$-linear but also $\Fqm$-linear. More precisely, according to~\cite{Gabidulin}, we consider \emph{vector rank-metric code} defined as follows.

\begin{definition}
  An \emph{$\Fqm$-vector rank-metric code} $\C$ is an $\Fqm$-linear subspace of $\Fqm^{n}$ endowed with the rank distance. Its dimension $k$ is  $\dim_{\Fqm}(\C)$ and its minimum rank distance is defined as
  \[
  \dr(\C) \eqdef
  \min_{{\bm c} \in \C \setminus \{\bm 0\}} \wr(\cv).
  \]
  A vector rank-metric code of dimension $k$ and minimum rank-distance $d$ is called an \emph{$[n,k,d]_{q^m}$ vector rank-metric code} (shortly \emph{vector code}).
\end{definition}

The next proposition clarifies  the correspondence between $\Fqm$-vector codes and matrix codes.
\begin{proposition}\cite[Proposition 1.5]{Gor21}\label{prop:vec_to_mat}
  If $\C \subseteq \Fqm^n$ is a vector rank-metric code of dimension $k$ over $\Fqm$ and $\Gamma$ is an $\Fq$-basis of $\Fqm$, then $\Gamma(\C)$ is an $\Fq$-linear rank-metric code of dimension $mk$ over $\Fq$.
\end{proposition}
Therefore an $[n,k,d]_{q^m}$ vector rank-metric code gives rise to an $[m\times n, mk, d]_q$ matrix rank-metric code.

\begin{remark}
  \label{rmk: diff dual}%
  Despite the correspondence described above, some caution is required regarding the duals of these codes~\cite{Gor21}.
  The dual of an \emph{$\Fqm$-vector rank-metric code} $\C$ (with respect to the standard inner product of $\Fqm^n$) is thus defined as
  
  \begin{equation}
    \label{eq:dual_vector}%
    \C^\perp = \set{ \bm x =(x_1,\dots,x_n) \in \Fqm^n \;\big|\; \langle \bm x, \bm c \rangle \eqdef \sum_{i=1}^n x_ic_i= 0, \forall \bm c=(c_1,\dots,c_n) \in \C}.
  \end{equation}
  
  In general, the matrix code associated to $\C^\perp$ (with respect to the bilinear form on $\Fqm^n$ defined above) via the basis $\Gamma$ does not coincide with the dual of the matrix code $\Gamma$, i.e., $\Gamma(\C^\perp) \ne \Gamma(\C)^\perp$. However, Ravagnani~\cite[Theorem~21]{Ravagnani_Rank-metric_2015} proved that
  \[
  \Gamma(\C^\perp) = \Gamma'(\C)^\perp
  \]
  if $\Gamma =(\gamma_1,\dots, \gamma_m)$ and $\Gamma'=(\gamma_1',\dots, \gamma_m')$ are orthonormal bases of $\Fqm$ with respect to the trace bilinear form
  \[
  \begin{array}{rcl}
    \F_{q^m}\times \F_{q^m} & \to &\F_q        \\
    (x,y)                   & \mapsto &\Tr_{q^m/q}(xy),
  \end{array}
  \]
  that is $\Tr_{q^m/q}(\gamma_i \gamma_j')=\delta_{i,j}$ for $i,j\in [m]$. Here, $\Tr_{q^m/q}:\F_{q^m}\to\F_q$ denotes the trace map of the field extension $\F_{q^m}/\F_q$.
\end{remark}

\subsubsection{Supports}
As in the Hamming metric, the notion of \emph{support} of a code plays a fundamental role in the study of rank-metric codes. In particular, support spaces will be one of the main tools in our study of locality.
In this section we define the \emph{column} and \emph{row support} for linear maps in $\Hom(U,V)$. These notions rely on the definitions of adjoint maps recalled in Section~\ref{subsec : adjoint}, generalizing the notions of column and row support of a matrix code.

\begin{definition}[Column support]
  The \emph{column support} of $f \in \Hom(U,V)$ is the image of $f$:
  \[
  \colsupp(f) \eqdef \im(f) \subseteq V.
  \]
\end{definition}
As expected, the dimension of the column support of $f$ is the rank of $f$, that is $\wr(f)$.

The notion of row support of a linear map is defined by the image of its adjoint map. By Theorem~\ref{thm:properties_adjoint}, it can equivalently be described in terms of the orthogonal complement of its kernel.
\begin{definition}[Row support]
  If $f\in \Hom(U,V)$, the \emph{row support} of $f$ is the image of $f^\dagger$, that is
  \[
  \rowsupp(f) \eqdef \im(f^\dagger) = \ker(f)^{\perp}.
  \]
\end{definition}
Since $\rk(f^\dagger) =\rk(f),$ the dimension of the row support of $f$ is the rank of $f$.
\begin{remark}[Row and column supports of matrices]
  Let $f \in \Hom(U,V)$ and let $F\eqdef\Mat{\B_U}{\B_V}{f}$ be its matrix representation. Then the column and row supports correspond to the classical column and row supports of $F$.
  Indeed, the column support of $f$ is generated by vectors $f(u_1), \ldots, f(u_n)$, whose coordinates with respect to $\B_V$ are the columns of $F$.
  By Lemma~\ref{lem:adjoint_matrices}, the matrix of $f^\dagger$ with respect to the dual basis is $F^\top$. Therefore, the row support of $f$ is generated by columns of $F^\top$, \ie by the  rows of $F$.
\end{remark}

We observe that, differently from the row and column supports of a matrix codeword, which depend on a choice of basis $\B_U$ and $\B_V$, the row and column supports of a linear map are intrinsically defined. Therefore, they are invariant under changes of bases of $U$ and $V$.

The following lemma describes the behaviour of the row support (respectively the column support) under composition of linear maps.

\begin{lemma}\label{lem:supports_composition}%
  For every $f \in \Hom(U, V)$, for every $g \in \Hom(S,U)$, we have
  \begin{itemize}
  \item $\colsupp(f\circ g)\subseteq \colsupp(f)$ with equality if $g$ is surjective,
  \item $\rowsupp(f\circ g)\subseteq \rowsupp(g)$ with equality if $f$ is injective.
  \end{itemize}
\end{lemma}

\begin{proof}
  We have $\im(f\circ g)\subseteq\im(f)$, and  $\rowsupp(f\circ g)= \im((f\circ g)^\dagger) = \im(g^\dagger\circ f^\dagger) \subseteq \im(g^\dagger) = \rowsupp(g) $. If $g$ is surjective, $\im(f\circ g) = f(\im(g))=f(U)= \im(f)$. If $f$ is injective, then $f^\dagger$ is surjective (by Theorem~\ref{thm:properties_adjoint}), and $\rowsupp(f\circ g) =\im(g^\dagger\circ f^\dagger)=\im(g^\dagger)$.
\end{proof}

As in the Hamming metric, the support of a code is obtained by combining the supports of its codewords. This leads to the following notions of row and column supports for rank-metric codes.

\begin{definition}
  The \emph{support} of a code $\C \subseteq \Hom(U, V)$ is the sum of the supports of its codewords:
  \[
  \rowsupp(\C) \eqdef \sum_{f \in \C} \rowsupp(f)  \quad \text{ and }   \quad \colsupp(\C) \eqdef \sum_{f \in \C} \colsupp(f).
  \]
  If $\rowsupp(\C) \ne U$ (resp. $\colsupp(\C) \ne V$), we say that $\C$ is \emph{row-degenerated} (resp. \emph{column-degenerated}). The code $\C$ is \emph{nondegenerate} if $\rowsupp(\C) = U$ and $\colsupp(\C) = V$.
\end{definition}

From now on, we only consider nondegenerate codes.

\section{Puncturing and shortening}

In this section we reformulate the notions of \emph{puncturing} and \emph{shortening} of rank-metric codes, originally introduced for matrix codes (see~\cite{coveringradius, Neri_tensor, BorelloScottiWCC26}) to rank-metric codes in $\Hom(U, V)$. Finally, we present some basic properties, that will be used throughout the paper.

\subsection{Codes in \texorpdfstring{\(\Hom(U,V)\)}{Hom(U,V)}}

Given $S$ an $\Fq$-vector space and a map $g \in \Hom(S, U)$, we define
\[
\begin{array}{rrcl}
  \pi_g : & \Hom(U, V) & \to     & \Hom(S, V) \\
  ~       & f          & \mapsto & f \circ g
\end{array}
\]

\begin{definition}[Rank puncturing]
  \label{def:rank-puncturing}%
  Let $\C \subseteq \Hom(U,V)$ be a rank-metric code. Let $S$ be an $\Fq$-vector space of dimension $s \leqslant n$. Fix $g$ an injective map in $\Hom(S, U)$. The \emph{puncturing of $\C$ with respect to $g$} is the set
  \[
  \pi_g(\C) \eqdef \{f \circ g \mid f \in \C\} \subseteq \Hom(S, V).
  \]
\end{definition}

Throughout this paper, we will focus on a particular class of puncturing, i.e., the ones defined by canonical inclusions.

\begin{notation}\label{not:inclusion}
  If $S$ is a subspace of $U$, we use the notation $\pi_S$ to denote the puncturing map $\pi_{\iota_S}$, where $\iota_S \in \Hom(S, U)$ is the canonical injective map. In this case, given $f \in \Hom(U,V)$,  $\pi_{S}(f)$ is nothing but the restriction $f_{|S}$.
\end{notation}

The following remark shows that restricting to canonical inclusions is not a loss of generality up to code equivalence.

\begin{remark}\label{rmk:equivalent_puncturings}%
  Let $\C \subseteq \Hom(U, V)$, and $g \in \Hom(S, U)$ be an injective map. Up to equivalence, the code puncturing $\pi_g(\C)$ only depends on $\C$ and $\im(g)$. More precisely, if we denote $T \eqdef \im(g) \subseteq U$, then $\pi_g(\C)$ and $\pi_T(\C)$ are equivalent rank-metric codes. Indeed,
  \[
  \pi_g(\C) = \{ f \circ g \mid f \in \C \} = \{ f \circ \iota_T \circ \psi \mid f \in \C \} =  \pi_T(\C) \circ \psi
  \]
  where $\psi : S \to T$ is an isomorphism such that $g = \iota_T \circ \psi$.


\end{remark}

\begin{definition}[Rank shortening]
  \label{def:rank-shortening}%
  Let $\C \subseteq \Hom(U,V)$ be a rank-metric code. Let $S$ be an $\Fq$-vector space of dimension $s\leqslant n$ and $g\in \Hom(S,U)$ be an injective map. Let $\imh$ be a subspace of $U$ of dimension $n-s$ such that $\im(g) \oplus \imh = U$.
  The \emph{(domain) shortening of $\C$ with respect to $g$ and $\imh$} is the subspace of $\Hom(S, V)$ defined as
  $$
  \Short(\C,g,\imh) \eqdef \{f \circ g  \mid f \in \C,  f_{\mid \imh} = 0\} = \pi_g( \ker(\pi_{\imh}) \cap \C).
  $$
\end{definition}

\begin{remark}\label{rmk: punctShort_comparaisons}
  Puncturing and shortening have already been defined in the literature, both for matrix codes and  for vector codes~\cite{coveringradius, Neri19, BorelloScottiWCC26}. We now compare our definitions of puncturing and shortening mainly with the recent formulation of~\cite{BorelloScottiWCC26}, which is stated in the setting of matrix codes. For this comparison, we translate our definitions into the language of matrices.
  Let $g \in \Hom(S,U)$ be an injective map as in Definition~\ref{def:rank-puncturing} and let $B \eqdef \Mat{\B_S}{\B_U}{g} \in \Fq^{n \times s}$ be its matrix representation, where $s = \dim_{\Fq} S$ and $\B_S$ is a basis of $S$. Since $g$ is injective, $B$ has full rank. In matrix representation,
  \[
  \pi_g(\C) =\pi_B(\C)= \{MB \mid M\in \C\}.
  \] 
  Up to transposition, this coincides with the puncturing of~\cite{coveringradius}, or with the \emph{right puncturing} of~\cite{BorelloScottiWCC26}.
  Let now $\imh$ be a subspace as in Definition~\ref{def:rank-shortening}, so that $\im(g) \oplus \imh = U$. Choose a basis $\B_{\imh}$ of $\imh$, and let $B' \in \Fq^{n \times (n-s)}$ be the matrix whose $j$-th column is the coordinate vector of the $j$-th element of the basis $\B_{\imh}$, with respect to $\B_U$. The assumption $\im(g) \oplus \imh = U$ is equivalent to the matrix $(B \mid B')\in \Fq^{n \times n}$ being invertible. Therefore, in matrix representation,
  \[
  \Short(\C, g, \imh) = \Short(\C, B, B') = \{MB \mid M \in \C, MB'=0\}.
  \]
  Again, up to transposition, this is the same matrix operation as the shortening of~\cite{coveringradius}, and it corresponds to the \emph{right shortening} of~\cite{BorelloScottiWCC26}.
\end{remark}
The following lemma introduces some useful properties of puncturing and shortening. In particular, it describes how these operations affect the column and row supports, the minimum distance and the dimension of the code.

\begin{lemma}\label{lem:prop_punct}
  Let $\C \subseteq \Hom(U, V)$ be a rank-metric code. Let $g\in \Hom(S, U)$ be an injective map. Then the following properties hold.
  \begin{enumerate}[label=(\roman*)]
  \item\label{it:inclusion}  $\colsupp(\pi_g(\C))\subseteq \colsupp(\C)$ and $\rowsupp(\pi_g(\C))\subseteq \rowsupp(g)$;
  \item\label{it:dist_min1} if $\dim(S)+ d_R(\C) > n$, then $d_R(\pi_g(\C)) \leq d_R(\C)$.
  \item\label{it:dist_min_borne2} If $\pi_g(\C)\neq 0$, then $\dr(\C)-(n-\dim_{\Fq}S)\le \dr(\pi_g(\C)) \le \dim_{\Fq} S$. 
    
  \end{enumerate}
  Let $\imh \subseteq U$ be a subspace such that $\im(g) \oplus \imh = U$. Let $T$ be an $\Fq$-vector space and $h\in \Hom(T,U)$ be any injective map with $\im(h) = \imh$. Then,
  \begin{enumerate}[resume,label=(\roman*)]
  \item\label{it:somme_dim} $\dim_{\Fq} (\pi_g(\C)) + \dim_{\Fq}(\Short(\C,h,\im(g))) = \dim_{\Fq} (\C)$;
  \item\label{it:dist_min2} if $\Short(\C,g,\imh)\neq 0$, then $\dr (\Short(\C,g,\imh)) \ge \dr(\C)$.
  \end{enumerate}
\end{lemma}

\begin{proof} \ref{it:inclusion} follows from Lemma~\ref{lem:supports_composition}. If $f \circ g \neq 0$ for every nonzero $f \in \C$, then Lemma~\ref{lem:supports_composition} implies that $d_R(\pi_g(\C)) \leq d_R(\C)$. But $f \circ g=0$ if and only if $\im(g) \subseteq \ker(f)$, which implies that
  \[\dim(S) \leq n-w_R(f) \leq n - d_R(\C).\]
We thus proved~\ref{it:dist_min1} by contraposition. Now, let us prove~\ref{it:dist_min_borne2}. First, the upper bound $ \dr(\pi_g(\C)) \le \dim_{\Fq} S $ follows directly from Lemma~\ref{lem:supports_composition} and from the injectivity of $g$.
  We now show the lower bound. Let $f \circ g$ be a nonzero codeword of $\pi_g(\C)$, then $f\neq 0$ and so, $\wr(f) \ge \dr(\C)$. We recall that if we restrict a linear map to a subspace of the domain, its rank decreases by at most the codimension of that subspace. Thus, since $\im(f \circ g)=f(\im(g))$, we get $\dim_{\Fq}(f(\im(g)))\ge \wr(f)-(n-\dim_{\Fq}(\im(g)))$. So, we have $\wr(f\circ g)\geq \wr(f)-(n-\dim_{\Fq} S)$ for every nonzero codeword $f$ and we get the desired inequality.

  Let us prove~\ref{it:somme_dim}. By the rank-nullity theorem, we have
  \[
  \dim(\C) = \dim_{\Fq} (\pi_g(\C) )+ \dim_{\Fq} (\ker(\pi_g)  \cap \C)\,.
  \]
  Now notice that for any $f \in \Hom(U, V)$, if $f \circ g = 0$ and $f_{\mid \imh} = 0$, then $f = 0$, since $\im(g) \oplus \imh = U$. Hence,
  \[
  \dim_{\Fq}(\ker(\pi_g)  \cap \C) = \dim_{\Fq}(\pi_h(\ker(\pi_g)  \cap \C))
  \]
  leading to the desired result since $\Short(\C,h,\im(g)) = \pi_h(\ker(\pi_{\im(g)})  \cap \C)$ by definition.
  
  Now, it remains to prove~\ref{it:dist_min2}. Let $0 \neq \varphi \in \Short(\C, g, \imh)$. By definition, there exists $f\in \C$ such that $f_{\mid \imh} =0$ and $\varphi = f \circ g$. Since $\im(g) \oplus \imh = U$ and $f$ vanishes on $\imh$, then $\im(f)=f(\im(g))$. Now, since $f(\im(g))=\im(f \circ g)$, we have $\im(f)=\im(f\circ g)=\im(\varphi)$ and so $\wr(\varphi)=\wr(f)$. Since $\varphi \neq 0$, also $f \neq 0$ and $\wr(f) \geq \dr(\C)$, which proves the result.
\end{proof}

The minimum distance of a rank-metric code can also be related to the dimension of spaces along which puncturings do not make the code smaller.

\begin{proposition}\cite[Proposition~7.2]{Neri_tensor}
  \label{prop:equiv_neri}%
  Let $\C \subseteq \Hom(U,V)$ of dimension $k$ and minimum rank distance $\dr(\C)$. For any $2\leqslant d \leqslant \min(m,n)$, the following statements are equivalent.
  \begin{enumerate}
  \item $\dr(\C) \geqslant d$;
  \item for every $\Fq$-vector space $S$ of dimension $s\geqslant n-d+1$, and every injective map $g\in \Hom(S,U)$, the punctured code $\pi_g(\C)$ has dimension $k$.
  \end{enumerate}
\end{proposition}

A correspondence between code puncturing and shortening, well known in the Hamming metric~\cite[Theorem~1.5.7]{HuffmanP10}, was also proved in the context of rank-metrix matrix codes (\cite[Theorem~3.5]{coveringradius} or~\cite[Theorem~2.6]{BorelloScottiWCC26}). We reformulate these results in the context of linear maps.

\begin{proposition}\label{prop:dual_short_punt}%
  Let $S$ be an $\Fq$-vector space and let $\C \subseteq \Hom(U,V)$ be a rank-metric code. Let $g\in \Hom(S,U)$ be an injective map.
  Let $\widetilde{g}\in \Hom(S,U)$ be a map and let $\widetilde{\imh}$ be a subspace of $U$ such that $\dim_{\Fq} \widetilde{\imh}=\dim_{\Fq} U - \dim_{\Fq} S$. Assume that,
  \begin{equation}\label{eq: cond_dual_shortening}
    \left\{
    \begin{array}{l}
      g^\dagger \circ \widetilde{g} = \id_S, \\
      g^\dagger_{\mid \widetilde{\imh}}=0\\
    \end{array}
    \right.
  \end{equation}
  Then,
  \[
  \pi_g(\C)^\perp = \Short(\C^\perp, \widetilde{g},\widetilde{\imh}).
  \]
\end{proposition}

\begin{proof}
  By~\cite[Theorem~2.6]{BorelloScottiWCC26} (applied to the transposed code), we have 
  \[
  \pi_g(\C)^\perp = \set{\phi \circ \widetilde{g} \mid \phi \in \C^\perp, \rowsupp(\phi) \subseteq \colsupp(g)}.
  \]
  By Theorem~\ref{thm:properties_adjoint}, this reformulates as follows.
  \[
  \pi_g(\C)^\perp = \set{\phi \circ \widetilde{g} \mid \phi \in \C^\perp, \im(g)^\perp \subseteq \ker(\phi)}.
  \]
  We just have to check that $\widetilde{\imh}=\ker(g^\dagger)=\im(g)^\perp$. We first observe that $g^\dagger$ is surjective. By the rank-nullity theorem, $\dim \ker(g^\dagger) = \dim U - \dim S$.
  Since, $g^\dagger_{\mid \widetilde{\imh}}=0$, we have $\widetilde{\imh} \subseteq \ker(g^\dagger)$.
  Now, since $\dim_{\Fq}\widetilde{\imh}=\dim_{\Fq} U - \dim_{\Fq} S$, we get the desired equality.
\end{proof}
\begin{remark}
  Proposition~\ref{prop:dual_short_punt} is a reformulation, in the language of morphisms, of~\cite[Theorem~2.6]{BorelloScottiWCC26}. We observe that this result is slightly more general than~\cite[Theorem~3.5]{coveringradius}. Indeed, the duality statement of~\cite[Theorem~3.5]{coveringradius} relates the puncturing defined by an invertible matrix $A$ to the shortening defined by $(A^\top)^{-1}$.
  Informally speaking, this inverse transpose appears because the shortening on the dual side has to be defined with respect to the data dual to those used for puncturing. In our setting, these data are encoded by the adjoint of the puncturing map. More precisely, if $B$ is the matrix of the puncturing map $g$ (see Remark~\ref{rmk: punctShort_comparaisons}), then the conditions in~\eqref{eq: cond_dual_shortening} become
  \[
  B^\top \widetilde B = I_r
  \qquad\text{and}\qquad
  B^\top \widetilde B' = 0,
  \]
  where $\widetilde B$ is the matrix of $\widetilde g$, and the columns of $\widetilde B'$ span the subspace $\widetilde{\imh}$. Thus, the matrix $(A^\top)^{-1}$ in~\cite[Theorem~3.5]{coveringradius} gives one particular way of producing matrices $\widetilde B$ and $\widetilde B'$ satisfying these two relations.
\end{remark}

We conclude this section by showing that the matrix interpretation of the definitions of puncturing and shortening (see Remark~\ref{rmk: punctShort_comparaisons}) behave well under the usual passage from vector-codes in $\Fqm$ to matrix-codes in $\Fq$.
\begin{lemma} \label{lem: dim passage gamma poiconne}
  Let $\C\subseteq \Fqm^n$. Let $\Gamma$ be a $\Fq$-basis of $\Fqm$. For every $B\in \Fq^{n \times s}$, and for every $B'\in\F_{q}^{n\times (n-s)}$ such that $(B|B')$ is invertible, we have $$\Gamma(\pi_B(\C))=\pi_B(\Gamma(\C)),$$ and $$\Gamma(\Short(\C,B,B'))=\Short(\Gamma(\C),B,B') .$$
\end{lemma}

\begin{proof}
  Write $B = (b_1 \mid \dots \mid b_s)$. For every $\bm c\in \C$, we have
  \begin{align*}
    {\Gamma}({\bm c} B ) & =(\Gamma(\langle \bm c, b_1 \rangle), \dots, \Gamma(\langle \bm c, b_s \rangle ) ) \\ &= ({\Gamma}(\bm c)\cdot b_1,\dots, {\Gamma}(\bm c)\cdot b_s) \\ &= {\Gamma}(\bm c)B,
  \end{align*}
  which proves the first equality.
  For the second equality, remark that for every $\bm c \in \C$, since $\Gamma$ is an isomorphism, $\bm cB' = 0$ if and only if $\Gamma(\bm c)B'=0$.
\end{proof}

\section{Local recovery in rank metric}

In this section, we introduce and analyse the notion of local recovery for codes endowed with the rank metric. Our goal is to adapt locality constraints that are used in the Hamming metric to rank-dedicated operations and structures (support, puncturing, etc.).

\paragraph{Local recovery seen as partial but efficient erasure decoding} In an erasure-channel context, assume that $Y = C + E \in \Fq^{m \times n}$ is a received word, where $C \in \C \subseteq \Fq^{m \times n}$ is a rank-metric code and $E \in \Fq^{m \times n}$ is rank-$1$ erasure, that is, a matrix such that $W = \colspan(E)$ is known and has dimension $1$.

The problem of finding $C$ given $Y$ and $W$ is known as the erasure decoding problem. A typical way to solve it is to consider a matrix $A \in \Fq^{n \times (n-1)}$ of rank $n-1$ such that $EA = 0$ (the knowledge of $W$ is sufficient for this), and to compute $YA = CA + EA = CA$. If $\dim \pi_A(\C) = \dim \C$, then one can recover the codeword $C$ from $YA = CA = \pi_A(C)$, for example by using linear algebra.

The local recovery problem is somewhat similar to erasure decoding, with two essential modifications: (i) given $u \in \Fq^n \setminus \{ 0 \}$, it is only required to recover the vector $Cu$ and not $C$ entirely,  (ii) to do so, we only have access to the image of $C$ on a subspace of $\Fq^n$ of small dimension.

In particular, and in contrast with~\cite{Kadhe_et_al19}, we believe that the local recovery property must concern any nonzero element of the code domain ($U$ for morphisms, $\Fq^n$ for matrices), instead of only canonical vectors.

\subsection{Locality for rank-metric codes}

Recall that given a subspace $S \subseteq U$, the map $\pi_S$ denotes the puncturing associated to the canonical inclusion $\iota_S : S \hookrightarrow U$ (see Notation~\ref{not:inclusion}). We now define a new notion of locality with respect to the rank metric.

\begin{definition}
  Let $\C \subseteq \Hom(U,V)$ and $r \in [n]$. We say that $u \in U \setminus \{ 0 \}$ \emph{has (rank) locality $r$ in $\C$} if there exists $S \subseteq U$ of dimension at most $r$, such that $u \notin S$ and
  $$\dim_{\Fq} \big( \pi_S(\C)\big)=\dim_{\Fq} \big(\pi_{S \oplus \vect{u}{\Fq}}(\C)\big).$$
  In this case, the space $S$ is called a \emph{helper space} for $u$.
\end{definition}

As a consequence of Remark~\ref{rmk:equivalent_puncturings}, we can rewrite the notion of rank locality with more general puncturings.





\begin{lemma}\label{lem:redefinition-localite}%
  Let $\C \subseteq \Hom(U,V)$, and $r \in [n]$. We say that $u \in U \setminus \{ 0 \}$ \emph{has (rank) locality $r$ in $\C$} if there exists an $\Fq$-space $S$ of dimension at most $r$ and an injective map $g \in \Hom(S, U)$ such that $u\notin \im(g)$ and
  $$\dim_{\Fq} \big( \pi_g(\C)\big)=\dim_{\Fq} \big(\pi_{h}(\C)\big)$$
  where 
  \[ \begin{array}{rrcl}
    h:& S \times \Fq &\rightarrow & U\\
      & (x,\lambda) & \mapsto & g(x)+\lambda u.  
  \end{array}\]
\end{lemma}

\begin{definition}
  A code $\C \subseteq \Hom(U,V)$ is $r$-\emph{rank locally recoverable} if every $u \in U \setminus \{0\}$ has rank locality $r$ in $\C$.
\end{definition}

Proposition~\ref{prop:equiv_neri} gives a trivial upper bound on the locality parameter of a code.

\begin{proposition}
  Let $\C \subseteq \Hom(U, V)$ be a rank-metric code of minimum rank distance $d\geqslant 2$. Then $\C$ is $(n-d+1)$-rank locally recoverable.
\end{proposition}
\begin{proof}
  Let $u\in U \setminus \{0\}$. Let $S\subseteq U$ of dimension $n-d+1$, such that $u\notin S$. By Proposition~\ref{prop:equiv_neri}, $$\dim_{\Fq}(\pi_{ S}(\C))=\dim_{\Fq}(\C) =\dim_{\Fq}(\pi_{ S \oplus \vect{u}{\Fq}}(\C)),$$ and since $u\notin S$, $u$ has locality $n-d+1$ in $\C$.
\end{proof}
Lemma~\ref{lem:redefinition-localite} allows us to give an easy translation of the notion of locality in the context of matrix codes, where the matrix $B$ plays the role of the linear map $h$.

\begin{definition}[Locality for matrix codes]
  \label{def:localite-matrices}%
  Let $\C\subseteq \Fq^{m\times n}$, and $r \in [n]$.
  \begin{itemize}
  \item We say that $u \in \Fq^n \setminus \{ 0 \}$ has (rank) locality $r$ in $\C$ if there exists  $B \in \F_q^{n \times r'}$ of rank $r' \le r$ such that $u \notin \colsupp(B)$ and
    $$
    \dim_{\Fq} \big(\pi_B(\C)\big)=\dim_{\Fq} \big(\pi_{(B \mid u)}(\C)\big).
    $$
  \item The code $\C \subseteq \Fq^{m\times n}$ is $r$-\emph{rank locally recoverable} if every $u \in \Fq^n  \setminus \{ 0 \}$ has rank locality $r$ in $\C$.
  \end{itemize}
\end{definition}

\begin{example}
  \label{ex:exemple-tres-simple}%
  Let $m, n \ge 1$ and consider the matrix code
  \[
  \C = \{ ( M | M ) \mid M \in \Fq^{m \times n} \} \subseteq \Fq^{m \times 2n}
  \]
  of dimension $mn$ over $\Fq$. We claim this code has locality $r = 1$. Indeed, for all nonzero $u = (u_1, u_2) \in \Fq^{2n}$ (with $u_i \in \Fq^n$), we have
  \[
  C u = M (u_1 + u_2), \qquad \forall C = (M | M) \in \C. 
  \]
  As a consequence: 
  \begin{itemize}
  \item if $u_2 = 0$, then $C u = C (u_1, 0) = M u_1 = C(0, u_1)$ for all $C \in \C$. Denote by $B \in \Fq^{n \times 1}$ the column matrix associated to the vector $(0, u_1) \in \Fq^n$; then one can easily check that $\pi_{(B | u)}(\C)$ and $\pi_{B}(\C)$ have the same dimension over $\Fq$, meaning that $\colspan(B)  = \vect{(0, u_1)}{\Fq}$ is a helper space for $u$.
  \item if $u_2 \ne 0$, we also have $C u = C (u_1 + u_2, 0)$ for all $C \in \C$. If $u_1 + u_2 = 0$, then $C u = 0$ for all $C \in \C$, hence $u$ has locality $0$. Otherwise, denote by $B' \in \Fq^{n \times 1}$ the column matrix associated to the vector $(u_1 + u_2, 0) \in \Fq^n$. Then $\dim_{\Fq} \pi_{(B' | u)}(\C) = \dim_{\Fq} \pi_{B'}(\C)$ and $\vect{(u_1 + u_2, 0)}{\Fq}$ is a helper space for $u$.
  \end{itemize}
\end{example}

Our definition of locality in the rank metric actually differs from the one given by Kadhe~\emph{et~al.} in~\cite{Kadhe_et_al19}. We will demonstrate this distinction thanks to upcoming Proposition~\ref{prop:equiv_localite_dual}, that requires to firstly prove a few elementary facts. To this end, up to the end of the section, we denote by $(v_1, \dots, v_m)$ a basis of $V$.

\begin{notation}\label{nota:rank-one-operators}
  For every $u\in U$ and every $\ell \in [m]$, let $\varphi_\ell[u]\in\Hom(U,V)$ be the rank-one operator defined by
  \[
    {\varphi_\ell[u]}(x) \eqdef \beta_U(u,x) v_\ell, \qquad \forall x\in U.
    \]
\end{notation}

\begin{lemma}\label{lem:prop_varphi}
  For every $u\in U$ and every $\ell\in[m]$, we have
  \begin{enumerate}[label=(\roman*)]
  \item\label{it:adjoint} $(\varphi_\ell[u])^\dagger(v)=\beta_V(v_\ell,v)u$ for all $v\in V$,
  \item\label{it:rowsupp} $\rowsupp(\varphi_\ell[u]) = \vect{u}{\Fq}$,
  \item\label{it:trace} for every $f\in\Hom(U,V)$, $\langle f,\varphi_\ell[u]\rangle = \beta_V(f(u),v_\ell)$.
  \end{enumerate}
\end{lemma}
\begin{proof}

  Let $u\in U$ and $\ell\in[m]$. For every $x \in U$ and $v\in V$, we have
  \[
  \begin{aligned}
    \beta_U((\varphi_\ell[u])^\dagger(v),x)
    &= \beta_V(v, \varphi_\ell[u](x))
    = \beta_V\left( v, \beta_U(u,x) v_\ell\right)
    = \beta_U( u,x) \beta_V( v, v_\ell)
    \\&= \beta_U\left(\beta_V(v_\ell,v) u, x\right),
  \end{aligned}
  \]
  which proves the first item. The second item then follows. To prove the last one, applying~\ref{it:adjoint} we get that, for every $f\in\Hom(U,V)$ and $i \in [m]$,
  \[
  \begin{aligned}
    \beta_V\!\left(f((\varphi_\ell[u])^\dagger(v_i)),v_i\right)
    &= \beta_V\!\Big(f(\beta_V(v_\ell,v_i) u), v_i\Big)
    \\&= \beta_V(v_\ell,v_i) \cdot \beta_V(f(u), v_i)
    \\&=\begin{cases}
    \beta_V(f(u), v_i)& \text{if } i =\ell,\\
    0 &\text{otherwise.}
    \end{cases}
  \end{aligned}
  \]
  The last equality holds because the basis $\mathcal{B}_V$ is orthonormal. Then 
  \[
  \langle f,\varphi_\ell[u]\rangle = \Tr(f\circ (\varphi_\ell[u])^\dagger) = \sum_{i=1}^m \beta_V\left(f((\varphi_\ell[u])^\dagger(v_i)),v_i\right) =\beta_V(f(u),v_{\ell})
  \]
  which proves the third item.
\end{proof}

\begin{lemma} \label{lemme_decompo}%
  Let $S \subseteq U$ of dimension $r$.
  Let $f\in \Hom(U,V)$ such that $\rowsupp(f)\subseteq S$. For every basis $(s_1, \dots, s_r)$ of $S$, there exists $(\lambda^j_i)_{i\in [r], j\in [m]}\in \Fq^{m\times r}$ such that $$f = \sum_{j=1}^m \varphi_j\left[\sum_{i=1}^r\lambda_i^j s_i\right].$$
\end{lemma}
\begin{proof} Fix  $(s_1, \dots, s_r)$ a basis of $S$.
  By hypothesis, $\im(f^\dagger) \subseteq S$. Therefore, there exists $(\lambda^j_i)_{i\in [r], j\in [m]}\in \Fq^{m\times r}$ such that for every $j\in[m]$, $f^\dagger(v_j) =\sum_{i=1}^r \lambda_i^js_i$. This implies that for every $v\in V$,
  $$f^\dagger (v) = \sum_{j=1}^m \beta_V(v_j,v) f^\dagger(v_j) = \sum_{j=1}^m (\varphi_j[f^\dagger(v_j)])^\dagger (v).$$
  We deduce that
  \[
  f = \sum_{j=1}^m \varphi_j[f^\dagger(v_j)] = \sum_{j=1}^m \varphi_j\left[\sum_{i=1}^r\lambda_i^j s_i\right].
  \]
\end{proof}

\begin{proposition} \label{prop:equiv_localite_dual}
  Let $\C\subseteq \Hom(U,V)$.
  Let $u\in U\setminus \{ 0 \}$, and $S\subseteq U$ of dimension $r$, with $u\notin S$. The following are equivalent.
  \begin{enumerate}[label=(\roman*)]
  \item\label{it:eqdim} $u$ has locality $r$ with helper space $S$ (i.e., $\dim_{\Fq}(\pi_S(\C)) = \dim_{\Fq}(\pi_{S \oplus \vect{u}{\Fq} }(\C))$);
  \item \label{it:dependency} for any basis $\set{s_1, \dots, s_r}$ of $S$, there exists $\psi_1,\dots,\psi_r \in \End(V)$ such that for all $f \in \C$, $f(u)=\sum_{i=1}^r \psi_i( f(s_i))$;
  \item\label{it:dual}There exist $g_1, \dots, g_m \in \Hom(U,V)$ such that, for all $\ell \in [m]$, we have $\rowsupp(g_\ell)\subseteq S$  and  $$\varphi_\ell[u]+g_\ell\in \C^\perp$$
    where the maps $\varphi_\ell[u]$ are defined in Notation~\ref{nota:rank-one-operators}.
  \end{enumerate}
\end{proposition}

\begin{proof}
  We are going to prove that \ref{it:eqdim} $\Leftrightarrow$ \ref{it:dependency} $\Leftrightarrow$ \ref{it:dual}. Let us consider the map
  \begin{equation}\label{eq:ev_S}%
    \begin{array}{rrcl}
      \ev_S : &\pi_{S \oplus \vect{u}{\Fq}}(\C) & \rightarrow & V^r                   \\
      ~ & f                                  & \mapsto     & (f(s_1),\dots,f(s_r))
    \end{array}
  \end{equation}
  By definition, the rank of $\ev_S$ is $\dim_{\Fq} \pi_S(\C)$. Then, by the rank-nullity theorem, \ref{it:eqdim} holds if and only if the map $\ev_S$ is injective.
  Assume~\ref{it:eqdim} holds. The map
  \[
  \begin{array}{rrcl}
    \ev_u : & \pi_{S\oplus \vect{u}{\Fq}}(\C) & \rightarrow & V    \\
    ~ &f                                   & \mapsto     & f(u)
  \end{array}
  \]
  satisfies that $\ker \ev_S = \set{0} \subseteq \ker \ev_u$. The fundamental theorem on homomorphisms ensures the existence of a map $\Psi : V^r \rightarrow V$ such that $\ev_u=\Psi \circ \ev_S$. In other words, for all $f \in \C$, $f(u)=\Psi(f(s_1),\dots,f(s_r))$, which implies~\ref{it:dependency}.
  \medskip
  Conversely, assume~\ref{it:dependency} holds. Then, for all $f \in \C$, we have $f(u)=\sum_{i=1}^r \psi_i(f(s_i))$. This implies that $\ev_S$ is injective, which implies~\ref{it:eqdim}.

  \medskip

  Let us prove that~\ref{it:dependency} implies~\ref{it:dual}. By~\ref{it:dependency}, for all $f \in \C$, we have $f(u)=\sum_{i=1}^r \psi_i( f(s_i))$.  Applying $\beta_V(\cdot, v_l)$  to this equality, Lemma~\ref{lem:prop_varphi}~\ref{it:trace} gives
  \begin{align*}
    \langle f, \varphi_\ell[u] \rangle & = \sum_{i=1}^r \Tr(\psi_i \circ f \circ \varphi_\ell[s_i]^\dagger) \\
    &= \sum_{i=1}^r \Tr( f \circ \varphi_\ell[s_i]^\dagger \circ \psi_i)\\
    &=\sum_{i=1}^r \Tr( f \circ(\psi_i^\dagger\circ \varphi_\ell[s_i])^\dagger) \\
    &= \sum_{i=1}^r \langle f, \psi_i^\dagger\circ \varphi_\ell[s_i]\rangle     \\
    &=  \langle  f, \sum_{i=1}^r  \psi_i^\dagger\circ \varphi_\ell[s_i]\rangle.
  \end{align*}
  Set $g_\ell \eqdef -\sum_{i=1}^r  \psi_i^\dagger\circ \varphi_\ell[s_i]$. By the previous computation, we have $\varphi_\ell[u] + g_\ell \in \C^\perp$. Moreover for every $i\in [r]$,  $\rowsupp( \psi_i^\dagger \circ \varphi_\ell[s_i]) \subseteq \rowsupp(\varphi_\ell[s_i]) \subseteq \vect{s_i}{\Fq}$ (by Lemma~\ref{lem:prop_varphi}~\ref{it:rowsupp} and Lemma~\ref{lem:supports_composition}), hence
  \[\rowsupp(\sum_{i=1}^r  \psi_i^\dagger\circ \varphi_\ell[s_i] )\subseteq \sum_{i=1}^r  \vect{s_i}{\Fq}=S.\]

  \medskip
  Now, it remains to prove that~\ref{it:dual} implies~\ref{it:dependency}.
  Denote $s_1, \dots, s_r$ a basis of $S$. Let $\ell \in [m]$. By Lemma~\ref{lemme_decompo}, there exists $(\lambda^j_i(\ell))_{i\in [r], j\in [m]}\in \Fq^{m\times r}$ such that $g_{\ell} = \sum_{j=1}^m \varphi_j[\sum_{i=1}^r\lambda_i^j(\ell) s_i].$
  This implies that for all $f\in \C$,
  \begin{equation}\label{eq_iii_ii}
    \langle f, \varphi_{\ell}[u] \rangle = \langle f , -g_{\ell}\rangle = \sum_{i=1}^r \sum_{j=1}^m -\lambda_i^j(\ell)\langle f ,   \varphi_j[s_i]\rangle.
  \end{equation}
  By Lemma~\ref{lem:prop_varphi}~\ref{it:trace},  $f(u)= \sum_{\ell=1}^m \beta_V(f(u),v_{\ell})v_{\ell} = \sum_{\ell=1}^m \langle f, \varphi_{\ell} [u]\rangle v_{\ell}$. By Equation~\eqref{eq_iii_ii}  we have
  \begin{align*}
    f(u) & = \sum_{\ell=1}^m \Big(\sum_{i=1}^r \sum_{j=1}^m -\lambda_i^j(\ell)\langle f ,   \varphi_j[s_i]\rangle\Big) v_{\ell} \\
    & = \sum_{i=1}^r \sum_{\ell=1}^m \Big( \sum_{j=1}^m -\lambda_i^j(\ell) \beta_V(f(s_i),v_j)\Big) v_{\ell}               \\ &= \sum_{i=1}^r \psi_i(f(s_i)),
  \end{align*} where we define $\psi_i\in \End(V)$ by $\psi_i:x\mapsto \sum_{\ell=1}^m \Big( \sum_{j=1}^m -\lambda_i^j(\ell) \beta_V(x,v_j)\Big) v_{\ell}. $
\end{proof}

It will be useful to also consider the translation in the matrix setting of the characterization in Proposition~\ref{prop:equiv_localite_dual}. Given a code $\C \subseteq \Fq^{m \times n}$, the following are equivalent:
\begin{enumerate}
\item a vector $u \in \Fq^n$ has locality $r$ in $\C$
\item there exist $P_1, \dots, P_r \in \Fq^{m \times m}$ and $s_1, \dots, s_r \in \Fq^n$ such that $u \notin \vect{s_1, \dots, s_r}{\Fq}$, and
  \[
  C u = \sum_{i=1}^r P_i C s_i, \qquad \forall C \in \C.
  \]
\item there exist $G_1, \dots, G_m \in \Fq^{m \times n}$, such that $\dim \sum_{\ell = 1}^m \rowsupp(G_\ell) \le r$ and
  \[
  E_\ell(u) + G_\ell \in \C^\perp, \qquad \forall  \ell \in [m],
  \]
  where $E_\ell(u)$ is the matrix with zeroes everywhere except for the $\ell$-th row which is $u$.
\end{enumerate}

\subsection{Constructing new LRCs from others}

It is clear that, if $\C, \C' \subseteq \Hom(U, V)$ are two $r$-rank LRCs, then $\C \cap \C'$ is also $r$-rank locally recoverable. Let us analyse other classical constructions of LRCs from others.

\begin{definition}
  Let $f^{(1)} \in \Hom(U^{(1)},V^{(1)})$ and $f^{(2)} \in \Hom(U^{(2)},V^{(2)})$. We define the following operations.
  \begin{itemize}
  \item \emph{direct sum:}
    \[
    \begin{array}{rrcl}
      f^{(1)}\oplus f^{(2)}:  & U^{(1)}\times U^{(2)} & \rightarrow & V^{(1)}\times V^{(2)} \\
      & u=(u^{(1)}, u^{(2)}) & \mapsto & \left(f^{(1)}(u^{(1)}),f^{(2)}(u^{(2)})\right).    
    \end{array}
    \]
  \item \emph{pairing:} if $U^{(1)}=U^{(2)}=U$,
    \[
    \begin{array}{rrcl}
      (f^{(1)}, f^{(2)}):  & U & \rightarrow & V^{(1)}\times V^{(2)} \\
      & u & \mapsto & \left(f^{(1)}(u),f^{(2)}(u)\right).    
    \end{array}
    \]
  \item \emph{copairing:} if $V^{(1)}=V^{(2)}=V$,
    \[
    \begin{array}{rrcl}
      f^{(1)} + f^{(2)}:  & U^{(1)}\times U^{(2)} & \rightarrow & V\\
      & u=(u^{(1)}, u^{(2)}) & \mapsto & f^{(1)}(u^{(1)}) + f^{(2)}(u^{(2)}).    
    \end{array}
    \]
  \end{itemize}
\end{definition}

\begin{remark}
  Given $M^{(1)}$ and $M^{(2)}$ the matrices of $f^{(1)}$ and $f^{(2)}$ in some bases, then $f^{(1)}\oplus f^{(2)}$, $(f^{(1)}, f^{(2)})$ and $f^{(1)} + f^{(2)}$ are respectively represented by the block matrices
  \[ \begin{pmatrix}
    M^{(1)} & 0 \\ 0 & M^{(2)}
  \end{pmatrix}, \quad 
  \begin{pmatrix}
    M^{(1)} \\ M^{(2)}
  \end{pmatrix} \text{ and }
  \begin{pmatrix}
    M^{(1)} & M^{(2)}
  \end{pmatrix}. \]
\end{remark}

\begin{proposition}\label{prop:cartesian_product_codes}
  Let $\C^{(1)} \subseteq \Hom(U^{(1)},V^{(1)})$ and $\C^{(2)} \subseteq \Hom(U^{(2)},V^{(2)})$ be two rank LRCs with locality $r^{(1)}$ and $r^{(2)}$ respectively.
  \begin{enumerate}[label=(\Roman*)]
  \item\label{it:direct_sum_codes} The direct sum of the codes 
    \[ \C^{(1)} \oplus \C^{(2)} \eqdef \set{f^{(1)}\oplus f^{(2)} \mid f^{(1)} \in \C^{(1)}, \: f^{(2)} \in \C^{(2)}} \subseteq \Hom\left(U^{(1)}\times U^{(2)},V^{(1)}\times V^{(2)}\right)\]
    is a rank LRC of locality $\max\{ r^{(1)}, r^{(2)} \}$.
  \item\label{it:copairing_codes} If $V^{(1)}=V^{(2)}=V$, the copairing of the codes 
    \[ \C^{(1)} + \C^{(2)} \eqdef \set{f^{(1)}+ f^{(2)} \mid f^{(1)} \in \C^{(1)}, \: f^{(2)} \in \C^{(2)}} \subseteq \Hom\left(U^{(1)}\times U^{(2)},V\right)\]
    is a rank LRC of locality $r^{(1)} + r^{(2)}$.

  \end{enumerate}
\end{proposition}

\begin{proof}
  Let us use the characterization~\ref{it:dependency} of Proposition~\ref{prop:equiv_localite_dual} to prove the statement. Let $j \in \set{1,2}$. For every nonzero $u^{(j)} \in U^{(j)}$, there exists $S^{(j)} = \vect{s_1^{(j)},\dots,s_{r^{(j)}}^{(j)}}{\Fq} \subseteq U^{(j)}$ not containing $u^{(j)}$ and $\psi^{(j)}_1,\dots,\psi_{r^{(j)}}^{(j)} \in \End(V^{(j)})$ such that for every $f^{(j)} \in \C^{(j)}$, we have
  \begin{equation}\label{eq: prop4.10_for_product}
  		f^{(j)}(u^{(j)})=\sum_{i=1}^{r^{(j)}} \psi_i^{(j)} f^{(j)}(s_i^{(j)}).
  \end{equation}
  We assume without loss of generality that $r^{(1)} \leq r^{(2)}$. If $r^{(1)} < r^{(2)}$, we set $s_i^{(1)}\eqdef 0$ and $\psi_i^{(1)}\eqdef \id_{V^{(1)}}$ for $r^{(1)}<i \leq r^{(2)}$.

  Let us prove~\ref{it:direct_sum_codes}.
  Let us take a nonzero $u=\left(u^{(1)},u^{(2)}\right) \in U^{(1)}\times U^{(2)}$. At least one of the $u^{(j)}$ is nonzero. If both are nonzero, then
  \[\begin{aligned}
  f^{(1)} \oplus f^{(2)}(u)= \left(f^{(1)}(u^{(1)}),f^{(2)}(u^{(2)})\right)&=
  \left(\sum_{i=1}^{r^{(1)}} \psi_i^{(1)} f^{(1)}(s_i^{(1)}), \sum_{i=1}^{r^{(2)}} \psi_i^{(2)} f^{(2)}(s_i^{(2)})\right)\\
  &=  \sum_{i=1}^{r^{(2)}}\psi_i  \left(f^{(1)} \oplus f^{(2)} \right)(s_i)
  \end{aligned}\]
  where $s_i=\left(s_i^{(1)},s_i^{(2)}\right)$ and $\psi_i=\psi_i^{(1)}\oplus \psi_i^{(2)}$. Then the set $S=\vect{s_1,\dots,s_{r^{(2)}}}{\Fq}$ does not contain $u$ and is a helper space for $u$ in $\C^{(1)} \oplus \C^{(2)}$. If $u^{(2)}=0$ (resp. $u^{(1)}=0$), then one can easily check that $S^{(1)}\times \set{0}$ (resp. $\set{0} \times S^{(2)}$) is a helper space for $u=(u^{(1)},0)$ (resp. $u=(0,u^{(2)})$) in $\C^{(1)} \oplus \C^{(2)}$. This completes the proof for~\ref{it:direct_sum_codes}.

  For the case~\ref{it:copairing_codes}, take $u=(u^{(1)}, u^{(2)})$ be a nonzero vector of $U^{(1)} \times U^{(2)}$ such that $u^{(1)} \neq 0$ and $u^{(2)} \neq 0$. The space $S\subseteq U^{(1)} \times U^{(2)}$ generated by the vectors $(s_1^{(1)}, 0), \dots, (s_{r^{(1)}}^{(1)}, 0), (0, s_1^{(2)}), \ldots, (0, s^{(2)}_{r^{(2)}})$ has dimension $r^{(1)}+r^{(2)}$ and does not contain $u$. Moreover, for every $ f^{(1)} \in \C^{(1)}$ and $f^{(2)} \in \C^{(2)}$, we have
  \begin{align*}
  	(f^{(1)}+f^{(2)})(u^{(1)}, u^{(2)})&=f^{(1)}(u^{(1)})+f^{(2)}(u^{(2)})\\
  															&=\sum_{i=1}^{r^{(1)}} \psi_i^{(1)} f^{(1)}(s_i^{(1)})+ \sum_{i=1}^{r^{(2)}} \psi_i^{(2)} f^{(2)}(s_i^{(2)})\\
  															&=\sum_{i=1}^{r^{(1)}} \psi_i^{(1)}(( f^{(1)}+f^{(2)})(s_i^{(1)},0))+ \sum_{i=1}^{r^{(2)}} \psi_i^{(2)}((f^{(1)}+ f^{(2)})(0, s_i^{(2)})).
  \end{align*}
  So, we can conclude that $S$ is a helper space for $u$.
  In conclusion, if $u^{(1)}=0$ and $u^{(2)}\neq 0$, we may take $S=\{0\} \times S^{(2)}$, where $S^{(2)}$ is a helper space for $u^{(2)}$ in $\C^{(2)}$. This is a helper space for $u$ in the copairing. The same holds for $u^{(2)}=0$ and $u^{(1)}\neq 0$.
\end{proof}

\begin{remark}
  If $U^{(1)}=U^{(2)}=U$, the pairing of the two LRCs $\C^{(1)}$ and $\C^{(2)}$, defined by
    \[ \left(\C^{(1)}, \C^{(2)}\right) \eqdef \set{\left(f^{(1)},f^{(2)}\right) \mid f^{(1)} \in \C^{(1)}, \: f^{(2)} \in \C^{(2)}} \subseteq \Hom\left(U,V^{(1)}\times V^{(2)}\right)\]
  is not necessarily a LRC. It depends on the way the helper spaces $S^{(1)}$ and $S^{(2)}$ for a same $u$ in $\C^{(1)}$ and $\C^{(2)}$ interact. If $u$ does not lie in $S\eqdef S^{(1)}+S^{(2)}$, then similar computations as for the proof of \ref{it:copairing_codes} ensure that $S$ is a helper space for $u$ in $\left(\C^{(1)}, \C^{(2)}\right)$. In particular, this holds when $\C^{(1)}=\C^{(2)}=\C$: if $\C$ is an $r$-rank LRC, then so is $(\C,\C)$.

  This situation is not surprising when compared with the situation in the Hamming metric. The pairing $\left(\C^{(1)}, \C^{(2)}\right)$ consists in \textquote{stacking} codewords of $\C^{(1)}$ and $\C^{(2)}$, which corresponds to \emph{interleaving}. In the Hamming metric, the interleaving of two different LRCs is not necessarily locally recoverable.
\end{remark}

\subsection{Comparison with the definition of Kadhe~\emph{et~al.}}\label{subsec:comparison}

Kadhe~\emph{et~al.} proposed in~\cite[Definition~2]{Kadhe_et_al19} another definition of locality in the rank metric. They consider that a code $\C \subseteq \Fq^{m \times n}$ is $r$-locally recoverable if every column of a matrix-codeword $C \in \C$ can be recovered by accessing at most $r$ other columns. Considering this code in $\Hom(\Fq^n,\Fq^m)$, this notion depends on the choice of bases for $\Fq^n$ and $\Fq^m$.

Our new definition of locality is coordinate-free and allows efficient recovery of \emph{any} element in the image of the matrix-codeword, and not only columns of the codeword\footnote{Said differently, a code is locally recoverable if we can recover any element of rank $1$ in the support of the codeword, given access to at most $r$ independent elements in its support, as it is the case for Hamming-LRCs.}. Both definitions thus differ in two ways: our definition demands that more elements should be recoverable (not only columns of the codewords), but allows much diverse helper spaces (neither only subsets of columns of the codewords).

The dependency  on the choice of bases in the definition of locality in~\cite{Kadhe_et_al19} induces that two equivalent codes may not share the same locality. Indeed, consider the code
\[
\C = \left\{ \begin{bmatrix} a & 0 & b & 0 \\ 0 & a & 0 & b \end{bmatrix} \mid (a,b) \in \Fq \right\}\,.
\]
It is easy to check that $\C$ has locality $1$ according to the definition of locality in~\cite{Kadhe_et_al19}. However, $\C$ is equivalent to
\begin{equation}\label{eq:ex_code}
  \C' = \left\{ \begin{bmatrix} a & 0 & a + b & 0 \\ 0 & a & 0 & b \end{bmatrix} \mid (a,b) \in \Fq \right\}
\end{equation}
whose third column cannot be recovered with only one other column. Hence $\C'$ has locality at least $2$ (actually, exactly $2$) according to~\cite{Kadhe_et_al19}.

On the contrary, according to our definition and to Proposition~\ref{prop:equiv_localite_dual}, locality is invariant under left-right-equivalence of codes.

\begin{lemma}\label{lem:localite-et-equivalence}%
  Let $\C \in \Hom(U, V)$ and $\C' \in \Hom(U', V')$ be rank-metric codes such that $\C = \beta \circ \C' \circ \alpha$ with $\alpha : U \isomto U'$ and $\beta : V' \isomto V$. Then, the codes $\C$ and $\C'$ share the same locality.
\end{lemma}

\begin{proof}
  Assume that $\C$ has locality $r$, and consider $u' \in U'$. Since $u \eqdef \alpha^{-1}(u')$ has locality $r$, there exist a helper space $S = \vect{s_1, \dots, s_r}{\Fq} \subseteq U$ not containing $u$, and homomorphisms $\psi_1, \dots, \psi_r \in \End(V)$, such that for all $f \in \C$, $f(u) = \sum_{i=1}^r \psi_i(f(s_i))$.

  We now claim that $S' \eqdef \alpha(S)$ is a helper space for $u'$ in $\C'$. Indeed, we have $u' \notin S'$ and, for any $f' \in \C'$,
  \begin{align*}
    f'(u') & = \beta^{-1}\Big( \underbrace{(\beta \circ f' \circ \alpha)}_{\in \C}(\alpha^{-1}(u'))\Big) \\
    & = \beta^{-1}\Big(\sum_{i=1}^r \psi_i((\beta \circ f' \circ \alpha)(s_i)) \Big) \\
    & = \sum_{i=1}^r \underbrace{(\beta^{-1} \circ \psi_i \circ \beta)}_{\in \End(V')}(f'(\alpha(s_i)))\,.
  \end{align*}
\end{proof} 

Lemma~\ref{lem:localite-et-equivalence} allows us to enlarge Example~\ref{ex:exemple-tres-simple} to a broader family of codes.

\begin{example}\label{ex:famille-localite-1}%
  Let $m, n \ge 1$ and fix $A \in \GL_m(\Fq)$. Then the matrix code
  \[
  \C_A = \{ ( M | AM ) \mid M \in \Fq^{m \times n} \} \subseteq \Fq^{m \times 2n}
  \]
  has locality $1$. Indeed, $\C_A$ is equivalent to the code $\C$ defined in Example~\ref{ex:exemple-tres-simple}, through the $\Fq$-isometry of $\Fq^{m \times 2n}$ given by: 
  \[
  C = (C_1 | C_2 ) \mapsto (C_1 | A^{-1} C_2)\,.
  \]
  We will see later that some of the codes in this family are remarkable, since they achieve a Singleton-like bound, see Section~\ref{sec:singleton-like-bound}.
\end{example}

Let us finally point out two examples that illustrate that the definition of locality in~\cite{Kadhe_et_al19} is not equivalent to ours.

\begin{example}
  Let us prove that $\C'$ defined in Equation~\eqref{eq:ex_code} has locality $1$ for our definition, whereas it has locality $2$ for~\cite{Kadhe_et_al19} (from previous discussion). Since $\C'$ and $\C$ are equivalent, according to Lemma~\ref{lem:localite-et-equivalence}, we only have to prove that $\C$ (also defined above) has locality $1$.

  Consider a nonzero $u=(u_1,u_2,u_3,u_4) \in \Fq^4$. For any codeword $C = \begin{bmatrix} a & 0 & b & 0 \\ 0 & a & 0 & b \end{bmatrix} \in \C$, we have $C u = \begin{bmatrix} a u_1 + b u_3 \\ a u_2 + b u_4 \end{bmatrix}$. Now, define $v = (u_2, u_1, u_4, u_3 )$. We can easily check that:
  \[
  C u = \begin{bmatrix} 0 & 1 \\ 1 & 0 \end{bmatrix} C v\,.
  \]
  Therefore, if $u$ and $v$ are not collinear, then $C u$ can be recovered by querying only $Cv$ and applying a linear map to it, proving that $u$ has locality $1$ according to Proposition~\ref{prop:equiv_localite_dual}.

  If $u$ and $v$ are collinear, this means that $u$ can be written as $u = [ u_1, \lambda u_1, u_3, \lambda u_3]$ with $\lambda \in \Fq$ such that $\lambda^2 = 1$. But in that case,
  \[
  C u =  \begin{bmatrix} a u_1 + b u_3 \\ \lambda(a u_1 + b u_3) \end{bmatrix} =\begin{bmatrix} 1 & 0 \\ \lambda & 0 \end{bmatrix} C w
  \]
  with $w = (u_1, 0, u_3, 0) \notin \vect{u}{\Fq}$. This also means that $Cu$ can be recovered by a query of rank $1$.
\end{example}

\begin{example}
  Consider now the matrix code of dimension $3$ defined as
  \[
  \mathcal{D} = \left\{
  \begin{bmatrix}
    a & c & b & c \\
    c & b & c & a+c
  \end{bmatrix}
   \mid (a,b,c) \in \mathbb{F}^3_2
  \right\} \subseteq \mathbb{F}_2^{2 \times 4}\,.
  \]
  The code $\mathcal{D}$ has locality $1$ according to~\cite{Kadhe_et_al19}: column $1$ can be recovered by column $4$ (and conversely), and column $2$ can be recovered by column $3$ (and conversely).

  However, $\mathcal{D}$ is not $1$-rank locally recoverable according to our definition, since the vector $u = (1, 1, 0, 0) \in \mathbb{F}_2^4$ does not have locality $1$. Indeed, for any $C =   \begin{bmatrix}  a & c & b & c \\ c & b & c & a+c \end{bmatrix} \in \mathcal{D}$, we have
  \[
  Cu = \begin{bmatrix} a + c \\ b + c\end{bmatrix}
    \]
    and the $15$ other evaluations of $C v$, for $v \in \mathbb{F}_2^4 \setminus \{ u \}$, are:
    \[
    \small
    \begin{aligned}
      &\begin{bmatrix} 0 \\ 0 \end{bmatrix},
      \begin{bmatrix} a \\ c \end{bmatrix},
      \begin{bmatrix} c \\ b \end{bmatrix},
      \begin{bmatrix} b \\ c \end{bmatrix},
      \begin{bmatrix} c \\ a+c \end{bmatrix},
      \begin{bmatrix} a+b \\ 0 \end{bmatrix},
      \begin{bmatrix} a+c \\ a \end{bmatrix},
      \begin{bmatrix} b+c \\ b+c \end{bmatrix},
      \begin{bmatrix} 0 \\ a+b+c \end{bmatrix},
      \\
      &\begin{bmatrix} b+c \\ a+b+c \end{bmatrix},
      \begin{bmatrix} a+b+c \\ b \end{bmatrix},
      \begin{bmatrix} a \\ a+b \end{bmatrix},
      \begin{bmatrix} a+b+c \\ a+c \end{bmatrix},
      \begin{bmatrix} b \\ a+b \end{bmatrix},
      \begin{bmatrix} a+b \\ a+b+c \end{bmatrix}.
    \end{aligned}
    \]
    One can check that none of them allows to recover $Cu$.
    
    Notice that, by Proposition~\ref{prop:equiv_localite_dual}~\ref{it:dependency}, the space generated by $u$ is a helper space for $v=(1,0,1,0)$ (since $Cv=\psi(Cu)$ with $\psi(x,y)=(x+y,0)$) and $v=(0,1,1,0)$ (with $\psi(x,y)=(y,y)$). Indeed, the morphism $\psi$ does not need to be an automorphism, which induces this asymmetry that cannot happen for locality $1$ in the Hamming metric (for nondegenerate codes).
\end{example}

\subsection{Relations with the Hamming metric}\label{subsec:relation_Hamming}

It is natural to question whether LRCs in the Hamming metric can help to design LRCs in the rank metric. One way to map Hamming-metric codes to rank-metric codes is to use the following \enquote{diagonal} construction. Assume $n \geq 2$, and let $\C_H \subseteq \Fq^n$ be a Hamming-metric code of dimension $k$ and minimum distance $d$. Let also $U$ be a vector space of dimension $n$ over $\Fq$, and denote by $(u_1, \dots, u_n)$ a basis of $U$. We then define
$$\C = \mathrm{Diag}(\C_H) \eqdef\set{f_\mathbf{c}  \mid \mathbf{c} \in \C_H}$$
where $f_\mathbf{c}(u_i) = c_i u_i$ for every $i \in [n]$. Then $\C  \subseteq \Hom(U,U)$ is a rank-metric code of dimension $k$ and minimum rank distance $d$.

\begin{lemma}\label{lem:locality_Hamming_to_rank}
  If $\C_H$ is locally recoverable with locality $r$, then $\C$ is $1$-rank locally recoverable.
\end{lemma}
\begin{proof}
  Let $i \in [n]$. Recall that since $\C_H$ is locally recoverable with locality $r$, there exist $\mathcal{R}_i \subseteq [n]$ of size at most $r$ such that $i \notin \mathcal{R}_i$ and some $(\alpha_{i,j})_{j \in \mathcal{R}_i} \in \Fq$ such that for every codeword $\mathbf{c}=(c_1, \dots, c_n) \in \C_H$, we have
  \begin{equation}\label{eq:relation_LRC}
    c_i = \sum_{j \in \mathcal{R}_i} \alpha_{i,j} c_j.
  \end{equation}

  Now, let us consider a general nonzero $u \in \Fq^n$ and write $u = \sum_{j=1}^n \lambda_j u_j$ with $\lambda_j \in \Fq$. Since $n \geq 2$, there exists $i \in [n]$ such that $w_i = \sum_{j \neq i} u_j$ is non-collinear to $u$. Let $\psi \in \End(U)$ such that
  \[
  \psi(u_j) = \lambda_j u_j +
  \begin{cases}
    \lambda_i \alpha_{i,j}  u_i & \text{if } j \in \mathcal{R}_i, \\
    0                           & \text{otherwise,}
  \end{cases}
  \]
  for $j \neq i$ and $\psi(u_i)=0$. Then, for every codeword $\mathbf{c}=(c_1, \dots, c_n) \in \C_H$, we have
  \[
  \begin{aligned}
    \psi \circ f_\mathbf{c}(w_i)
    &= \sum_{j\neq i} \psi \circ f_\mathbf{c}(u_j)
    = \sum_{j\neq i} c_j \psi(u_j)
    = \sum_{j\neq i} \lambda_j c_j u_j + \sum_{j\in \mathcal{R}_i} \lambda_i \alpha_{i,j} c_j u_i
    \\&= \sum_{j\neq i} \lambda_j c_j u_j + \lambda_i c_i u_i = f_\mathbf{c}(u),
  \end{aligned}
  \]
  which means that the 1-dimension space spanned by $w_i$ is a helper space for $u$ in $\C$ (Proposition~\ref{prop:equiv_localite_dual}), hence $u$ has locality $1$ in $\C$.
\end{proof}

The previous lemma clearly illustrates the difference between the pre-existing notion of local recoverability in the rank metric compared to the one we develop in the present work. We can see the proposition of Kadhe~\emph{et~al.} as a straightforward translation of the notion in the Hamming metric: if $\C_H$ is an LRC, then $\C$ is a rank LRC in the sense of~\cite[Definition~2]{Kadhe_et_al19} with the \emph{same locality}, in contrast with Lemma~\ref{lem:locality_Hamming_to_rank}.

\subsection{Information spaces for rank-metric codes and locality of MRD codes}

The goal of this section is to formalize the notion of information spaces of $\Fq$-linear rank-metric codes, in analogy with the notion of information set in Hamming metric. Recall that an information set for a code in the Hamming metric is a minimal set of indices such that every codeword is uniquely determined by its values at the corresponding positions. In other words, the restriction of the code on this set of positions has the same dimension as the code itself. In rank metric, we will translate this idea to subspaces of the code domain.

\begin{definition}
  Let $\C\subseteq\Hom(U,V)$ be a nonzero code. A subspace $S\subseteq U$ is said to be an \emph{information space} for $\C$ if
  \begin{enumerate}[label=(\roman*)]
  \item \label{item:dim-info-space} $\dim_{\Fq}( \pi_S(\C)) = \dim_{\Fq}(\C)$,
  \item \label{item:small-info-space} and for all $S' \varsubsetneq S$,  we have $\dim_{\Fq}( \pi_{S'}(\C)) < \dim_{\Fq}(\C)$.
  \end{enumerate}
\end{definition}

Notice that, given a code $\C \subseteq \Hom(U,V)$, the set of spaces satisfying~\ref{item:dim-info-space} is not empty (it contains $U$). So it admits an element of minimal dimension. Such an element is an information space of the code $\C$.

It is worth mentioning that, when dealing with the tensor representation of $k$-dimensional matrix codes, the vector space $\Fq^k$ is sometimes called the \emph{information space} (e.g., see~\cite[\S4]{Neri_tensor}). Here, the terminology is different: an information space for a $k$-dimensional code is a subspace of $U$ whose dimension (over $\Fq$) is at least $\frac{k}{m}$, where $m=\dim_{\Fq}(V)$.

The notion of information space is tightly related to the one of locality. If $S$ is (or contains) an information space for $\C \in \Hom(U,V)$ then $S$ is a helper space for any $u \notin S$. We thus use this notion to investigate the locality of MRD (Maximum Rank Distance) codes. We recall that MRD codes are the ones achieving the \emph{Singleton bound}, reformulated for homomorphisms below.

\begin{theorem}\cite[Theorem~3.5]{Gor21}
  \label{thm:morphisme_Singleton}
  Let $\C \subseteq \Hom(U,V)$ of dimension $k$  and minimum rank distance $d$. Then
  \[
  k \leq \max\{m, n\} (\min\{m, n\}-d+1).
  \]
\end{theorem}


For MRD codes, Proposition~\ref{prop:equiv_neri} translates into the following result.

\begin{proposition} \label{infospace_dist}
  Let $\C \subseteq \Hom(U,V)$ be an MRD code of dimension $k$ and minimum rank distance $d\geqslant 2$. Then, every $S\subseteq U$ of dimension $n-\min(m,n)+\frac{k}{\max(m,n)}$ contains an information space for $\C$.
\end{proposition}

\begin{proof}
  First notice that the dimension $k$ of an MRD code $\C$ is divisible by $\max\{m, n\}$.
  
  In the case where $m\geqslant n$, the Singleton bound gives $\frac{k}{m}=n-d+1$. Let $S\subseteq U$ of dimension $\frac{k}{m}$. By Proposition~\ref{prop:equiv_neri}, $\dim_{\Fq}(\pi_S(\C)) = \dim_{\Fq}(\C)$, and $S$ contains an information space for $\C$.
  
  In the other case ($n \geqslant m$), the Singleton bound becomes $\frac{k}{n}=m-d+1$. Let  $s=n-m+\frac{k}{n}$ and $S\subseteq U$ of dimension $s$. Since $s \geqslant n-d+1$, by Proposition~\ref{prop:equiv_neri},  $\dim_{\Fq}(\pi_S(\C)) = \dim_{\Fq}(\C)$, and $S$ contains an information space for $\C$.
\end{proof}

As a consequence, for any MRD code $\C \subseteq \Hom(U, V)$ with $m \ge n$, every subspace $S \subseteq U$ of dimension $\frac{k}{m}$ is an information space for $\C$.

\begin{corollary}
  Let $\C\subseteq \Hom(U,V)$ of dimension $k$.  If $m\geqslant n$ and if $\C$ is MRD, then $\C$ is $\frac{k}{m}$-rank locally recoverable. Besides, for every $u\in U$,  every $S\subseteq U$ of dimension $\frac{k}{m}$ such that $u\notin S$ is a helper space for $u$.
\end{corollary}

\begin{proposition}
  Let $\C\subseteq \Hom(U,V)$ of dimension $k$. If $\C$ is MRD, and $m\geqslant n$, then $\C$ cannot be locally recoverable with locality $r < \frac{k}{m}$.
\end{proposition}
\begin{proof}
  Since $\C$ is MRD, $\C^\perp$ is also MRD (see
  {\cite[Theorem~4.13]{Gor21}}) and the minimum distance of $\C^\perp$ is $\frac{k}{m}+1$. Let $u \in U$. If $u$ has rank-locality  $r$ in $\C$, then by Proposition~\ref{prop:equiv_localite_dual}, there exists $x \in \C^\perp$ such that $w_R(x)\leqslant r+1$. This implies that $\frac{k}{m}+1 \leqslant r+1$, and therefore $\frac{k}{m}\leqslant r$.
\end{proof}

\subsection{Locality for \texorpdfstring{$\Fqm$}{F\_\{q\^m\}}-linear codes}

In this section, we make explicit the definition of rank-locality in the context of $\Fqm$-linear codes and we give some results of the last section in this context.

\begin{definition}\label{def: localite vecteurs}%
  Let $\C\subseteq\F_{q^m}^n$ and $r \in [n]$.
  We say that $u \in \F_q^n \setminus \{0\}$ has (rank) locality $r$ in $\C$ if there exist $r' \leqslant r$ and $B\in \F_q^{n \times r'}$ of rank $r'$, with $u \notin \colsupp(B)$, such that
  $$\dim_{\Fqm} \big(\pi_B(\C)\big)=\dim_{\Fqm} \big(\pi_{(B\mid u)}(\C)\big).$$
  In this case, the space $S$ is called a \emph{helper space} for $u$.
\end{definition}

\begin{definition} \label{def lrc rank}
  A code $\C \subseteq \F_{q^m}^n$ is $r$-\emph{rank locally recoverable} if every $u\in \Fq^n\setminus \{0\}$ has rank locality $r$ in $\C$.
\end{definition}

\begin{proposition}
  Let $\C\subseteq \F_{q^m}^n$. If $\C$ is $r$-rank locally recoverable, then for every basis $\Gamma$ of $\Fqm/\Fq$, the code $\Gamma(\C) \subseteq \Fq^{m \times n}$ is  $r$-rank locally recoverable.
\end{proposition}

\begin{proof}
  If $\dim_{\Fqm}(\pi_B(\C))=t$, then $\dim_{\Fq}(\Gamma(\pi_B(\C))=mt$. Therefore, if $\dim_{\Fqm} (\pi_B(\C))= \dim_{\Fqm}(\pi_{(B\mid  u)}(\C))$, then $\dim_{\Fq} (\Gamma(\pi_B(\C)))= \dim_{\Fq}(\Gamma(\pi_{(B\mid  u)}(\C)))$, and by Lemma~\ref{lem: dim passage gamma poiconne}, $$\dim_{\Fq} (\pi_B(\Gamma(\C)))= \dim_{\Fq}(\pi_{(B\mid u)}(\Gamma(\C))),$$ and $\Gamma(\C)$ is $r$-rank locally recoverable.
\end{proof}

In the context of $\Fqm$-linear codes, the characterizations for local recoverability provided by Proposition~\ref{prop:equiv_localite_dual} become the following.
\begin{proposition} \label{equiv mot dual}
  Let $\C\subseteq\F_{q^m}^n$.
  Let also $u\in \Fq^n$, and $S\subseteq U$ of dimension $r$, with $u\notin S$. For every $B =(b_1, \dots, b_r)\in \Fq^{n\times r}$ such that $\colsupp(B)=S$, the following are equivalent.
  \begin{enumerate}[label=(\roman*)]
  \item\label{it:eqdimvect} $\dim_{\Fqm}(\pi_B(\C)) = \dim_{\Fqm}(\pi_{(B|u)}(\C))$
  \item \label{it:dependencyvect} There exist $\lambda_1, \dots, \lambda_r\in \F_{q^m}$ such that for every $ \bm c \in \C$, $\langle \bm  c,  u \rangle = \sum_{i=1}^r \lambda_i \langle  \bm  c, { b_i}\rangle.$ 
  \item\label{it:dualvect} There exists  $\bm  x \in \Fqm^n$ such that $\rowsupp(\bm  x)\subseteq S$ and  $u+ \bm x\in \C^\perp$.

  \end{enumerate}
\end{proposition}

\begin{proof}
  Since $\pi_{(B|u)}(\C) = \{(\langle \bm  c, {b_1} \rangle, \dots, \langle \bm  c, {b_r} \rangle,  \langle  \bm  c, u \rangle ), \bm  c \in \C\}$, the equivalence \ref{it:eqdimvect} $\Leftrightarrow$ \ref{it:dependencyvect} is clear.

  Suppose  now that there exist $\lambda_1, \dots, \lambda_r\in \F_{q^m}$ such that for every $\bm  c \in \C$, \[\langle \bm  c,  u \rangle = \sum_{i=1}^r \lambda_i \langle  \bm  c, { b_i}\rangle.\]
  Set ${\bm  x} :=   - \sum_{i=1}^r\lambda_i b_i.$
  Then $u + {\bm x}$ lies in $\C^\perp$ and $\rowsupp(\bm  x)\subseteq S$. This proves that~\ref{it:dependencyvect} implies~\ref{it:dualvect}.

  Conversely, suppose there exists $ \bm x\in \Fqm^n$ such that $\rowsupp(\bm  x)\subseteq S$ and  $u+\bm  x\in \C^\perp$. Then, there exist  $\lambda_1, \dots, \lambda_r\in \F_{q^m}$ such that $\bm  x = \sum_{i=1}^r \lambda_i b_i$. Since $u+\bm  x\in \C^\perp$, we obtain that for all $ \bm c\in \C, \langle \bm  c,  u \rangle = \sum_{i=1}^r \lambda_i \langle \bm   c, { b_i}\rangle.$ This proves that~\ref{it:dependencyvect} $\Leftrightarrow$~\ref{it:dualvect}.
\end{proof}

In the case of an $\Fqm$-vector code $\C$, the previous proposition implies that the endomorphisms $\psi_i$ from Proposition~\ref{prop:equiv_localite_dual}~\ref{it:dependencyvect} correspond to companion matrices (matrices of the multiplication by scalars in $\Fqm$) for the corresponding matrix code $\Gamma(\C)$. The third items of Propositions~\ref{prop:equiv_localite_dual} and~\ref{equiv mot dual} both give a characterization of local recoverability related to the dual codes. As already mentioned in Remark~\ref{rmk: diff dual}, we have $\Gamma(\C)^\perp\neq\Gamma(\C^\perp)$. Therefore, a  parity-check equation $\langle \bm c , \bm x \rangle =0 \in \Fqm$ does not necessarily lead to a  parity-check equation $\langle \Gamma(\bm c) , \Gamma(\bm x) \rangle =0 $. Since $\Fq^m$ has dimension $m$ over $\Fq$ an equation $\langle \bm c , \bm x \rangle =0 \in \Fqm$ gives rise to $m$ parity-check equations in $\Fq$. This explains why Proposition~\ref{prop:equiv_localite_dual}~\ref{it:dual} involves $m$ codewords in the dual of $\Gamma(\C)$, whereas~\ref{it:dualvect} of the previous proposition only relies on the existence of one vector (and all its $\Fq^m$-multiples).

Let us now prove that an $\Fqm$-linear MRD code of dimension $k$ has locality exactly $k$. To do so, we will study the notion of information space for $\Fqm$-linear codes.

\begin{lemma}\label{lem:info_space_dim_k}%
  Let $\C\subseteq\Fqm^n$ be an $\Fqm$-linear code of dimension $k$ (over $\Fqm$). There exists $B\in \Fq^{n\times k}$ of rank exactly $k$,  such that $\dim_{\Fqm}(\pi_B(\C))=\dim_{\Fqm}(\C)$.
\end{lemma}
\begin{proof}
  Let $G\in \Fqm^{k\times n}$ be a generator matrix of $\C$. Since $\rk(G)=k$, for every $A\in \GL_n(\Fq)$, $\rk(GA)=k$. Therefore, there exists a set $I=\{i_1, \dots, i_k\} \subseteq [n]$ such that the columns of $GA$ indexed by $I$ form an invertible matrix. Denote $A_I$ the matrix obtained from $A$ by keeping its columns indexed by $I$. We then have $\dim_{\Fqm}({\pi_B}(\C))=\dim_{\Fqm}(\C)$.
\end{proof}

Remark that if $\C \subseteq \Fqm^n$ has dimension $k$ over $\Fqm$ and if $\dim_{\Fq} \pi_B(\C) = \dim_{\Fq}(\C)$ for some $B \in \Fq^{n\times r}$ (with $r \le n$), then we must have $k \le r$ since $\pi_B(\C)\subseteq\Fqm^r$. This leads us to define information spaces for $\Fqm$-linear codes in the following way.

\begin{definition}\label{def:information_space_Fqm_linear}%
  Let $\C\subseteq\F_{q^m}^n$ with  $\dim_{\Fqm}(\C)= k$. A $k$-dimensional subspace $S\subseteq \F_q^n$ is called an \emph{information space for $\C$} if $\pi_B(\C)=\F_{q^m}^k$ for every  $B\in \F_q^{n\times k}$ whose columns span $S$.
\end{definition}

Notice that Lemma~\ref{lem:info_space_dim_k} shows that any $\Fqm$-linear code admits at least one information space.

We now recall the Singleton bound for vector codes, which follows directly from the matrix framework~\cite[Theorem~3.5]{Gor21}.

\begin{theorem}[Singleton bound for $\Fqm$-linear vector codes] An $[n,k,d]_{q^m}$ vector rank-metric code $\C$ satisfies
  \[
  k \leq  n-d+1.
  \]
  Codes whose parameters reach this bound are called \emph{Maximum Rank Distance (MRD) codes.}
\end{theorem}

The particular case of $\Fqm$-linear MRD codes  implies the following results thanks to Proposition~\ref{prop:equiv_neri}.

\begin{corollary}
  \label{coro:info_space_MRD}%
  Let $\C\subseteq \F_{q^m}^n$ be an MRD code of dimension $k$.
  Then every $S\subseteq \F_{q^m}^n$ of dimension $k$ is an information space for $\C$.
\end{corollary}

\begin{corollary}\label{coro:localite_MRD}
  Let $\C\subseteq \F_{q^m}^n$ be an MRD code of dimension $k$. Then, $\C$ is $k$-rank locally recoverable. Besides, for every $u\in \Fq^n$, every space $S$ of dimension $k$ such that $u\notin S$ is a helper space for $u$ in $\C$.

\end{corollary}

A last result finally shows that the dimension of helper spaces of MRD $\Fqm$-linear codes must be at least the dimension of the code.

\begin{proposition}
  Let $\C\subseteq \Fqm^n$ of dimension $k$ over $\Fqm$. If $\C$ is MRD, then $\C$ cannot be locally recoverable with locality strictly less than $k$.
\end{proposition}
\begin{proof}
  Since $\C$ is MRD, $\C^\perp$ is also MRD (see~\cite[Corollary~41]{Ravagnani_Rank-metric_2015}), and the minimum distance of $\C^\perp$ is $k+1$. Let $u\in \Fq^n$. If $u$ has rank locality  $r$ in $\C$, then by Proposition~\ref{equiv mot dual}, there exists $\bm x\in \C^\perp$ such that $w_R(\bm x)\leqslant r+1$. This implies that $k+1 \leqslant r+1$, and therefore $k\leqslant r$.
\end{proof}

\section{A Singleton-like bound for rank LRCs}
\label{sec:singleton-like-bound}%

This section aims at establishing a Singleton-like bound of rank-metric codes, which involves the locality and the three classical parameters. Such a bound was obtained in the Hamming metric in~\cite{locality12}.

Given subspaces $S, \imh \subseteq U$ such that $S \oplus \imh = U$, we denote by $\Short(\C, S,\imh)$ the shortened code $\Short(\C, \iota_S,\imh)$ using the canonical inclusion $\iota_{S}$ (see Notation~\ref{not:inclusion}).

\begin{lemma} \label{lemme_pour_singleton}
  Let $\C\subseteq \Hom(U,V)$ be an $r$-rank locally recoverable code. Let $S, \imh \subseteq U$ be subspaces such that $S \oplus S'=U$. Then, the code $\widetilde{\C} = \Short(\C, S, \imh)$ is $r$-rank locally recoverable.
\end{lemma}
\begin{proof}
  Let $u\in S \setminus \{ 0 \}$. Let $\psi \in \Hom(S,U)$ be an injective map such that $\psi^\dagger \circ \iota_{S} = \id_{S}$, and $\psi^\dagger \circ \iota_{S'} =0$.
  By Proposition~\ref{prop:dual_short_punt}, we have
  $$
  \widetilde {\C}^\perp = \Short(\C, S, S')^\perp = \pi_{\psi}(\C^\perp) .
  $$
  Since $\psi$ is injective, the map  $\psi^\dagger$ is surjective (Theorem~\ref{thm:properties_adjoint}). Consider $w\in U \setminus \{ 0 \}$ such that $\psi^\dagger(w)=u$.
  Since $\C$ is $r$-rank locally recoverable, $w$ has locality $r$ in $\C$, and by Proposition~\ref{prop:equiv_localite_dual}, there exist $W\subseteq U$ of dimension $\le r$ with $w\notin W$, and $f_1, \dots, f_m \in \Hom(U,V)$, such that for all $i\in [m]$, $\rowsupp(f_i)\subseteq W$ and $\varphi_i [w]+f_i\in \C^\perp.$ Since $\widetilde{\C}^\perp = \pi_\psi(\C^\perp)$, we get $(\varphi_i [w]+f_i)\circ \psi \in \widetilde{C}^\perp$ for every $i\in [m]$.

  Let now $i\in [m]$.
  Remark that $(\varphi_i [w]+f_i)\circ \psi = \varphi_i[\psi^\dagger(w)] + f_i\circ \psi = \varphi_i[u] +f_i\circ \psi$.
  Besides, $\rowsupp(f_i) = \im(f_i^\dagger) = \ker(f_i)^{\perp} \subseteq W$.  This implies that $W^{\perp}\subseteq \ker(f_i)$. Consider $\rho:U\to U/W^{\perp}$ the canonical surjection.
  Then, by the fundamental theorem on homomorphisms, there exists $g_i \in \Hom( U/W^{\perp}, V)$ such that $f_i = g_i \circ \rho$.

  We have $\rowsupp(f_i\circ \psi) = \rowsupp( g_i \circ \rho \circ \psi ) \subseteq  \rowsupp( \rho \circ \psi)$ by Lemma~\ref{lem:supports_composition}. Note that $\dim_{\Fq}(\rowsupp( \rho \circ \psi))=\dim_{\Fq}(\im( \rho \circ \psi))\leqslant \dim_{\Fq}(\im(\rho)) = r$.

  To summarize, the subspace $T = \rowsupp( \rho \circ \psi) \subseteq U$ has dimension at most $r$, and is such that for every $i\in [m]$, $\rowsupp(f_i\circ \psi)\subseteq T$.
  By Proposition~\ref{prop:equiv_localite_dual}, $\widetilde{\C}$ is $r$-rank locally recoverable.
\end{proof}

We now derive a Singleton-like upper bound for rank locally recoverable codes.
\begin{theorem}\label{thm:singleton-lrc}%
  Let $\C \subseteq \Hom(U, V)$ be a code of dimension $k$, and minimum rank distance $d$. Suppose $\C$ is $r$-rank locally recoverable. Then
  \begin{equation}
    \label{eq:singleton-lrc}%
    d \leqslant n-\left\lceil\frac{k}{m}\right\rceil+2 - \left\lceil\frac{k}{rm}\right\rceil.
  \end{equation}
  A code achieving this bound with equality
  is called an \emph{optimal rank locally recoverable code}.
\end{theorem}

\begin{proof}
  

  The proof works as follows. We build recursively a sequence of nonzero codes $\C^{(j)} \subseteq \Hom(U^{(j)},V)$ of dimension $k_j$ and minimum distance $d_j$ satisfying
  \begin{enumerate}[label=(\roman*)]
  \item\label{it:length}$n_j \eqdef \dim_{\Fq}(U^{(j)}) \geqslant n-j(r+1)$,
  \item\label{it:dim} $k_j\geqslant k-jrm$,
  \item\label{it:ineq} $n_j-\frac{k_j}{m}\leqslant n-\frac{k}{m}-j$,
  \item\label{it:dist} $d_j\geqslant d_{j-1}$,
  \item\label{it:LRC} $\C^{(j)}$ is $r$-rank locally recoverable.
  \end{enumerate}
  Then, we apply the classical Singleton bound on the last code we obtain, and we derive the expected result.  

  Set $\C^{(0)} = \C$ and $U^{(0)}=U$. Thus, $n_0 = n$, $k_0 = k$ and $d_0=d$.
  For $j\geqslant 1$, if we assume that $\C^{(j-1)}$ is a code satisfying hypotheses~\ref{it:length}--\ref{it:LRC}, we define $\C^{(j)}$ in the following way.
  \begin{itemize}
  \item Pick a nonzero $u^{(j)}\in U^{(j-1)}$.
  \item Since $\C^{(j-1)} \subseteq \Hom(U^{(j-1)},V)$ is $r$-rank locally recoverable, there exists $W_j\subseteq U^{(j-1)}$ of dimension $r_j \le r$, such that $u^{(j)}\notin W_j$ and
    \begin{equation}\label{eq:locality_j-1}
      \dim_{\Fq}\left(\pi_{W_j}(\C^{(j-1)})\right) = \dim_{\Fq}\left(\pi_{W_j\oplus \vect{u^{(j)}}{\Fq}} (\C^{(j-1)})\right) .
    \end{equation}
  \item  Finally, set  $\C^{(j)} \eqdef \text{Short}\left(\C^{(j-1)},W_j', W_j\oplus \vect{u^{(j)}}{\Fq}\right)$, for some subspace $W_j'\subseteq U^{(j-1)}$ such that $W'_j \oplus W_j\oplus \vect{u^{(j)}}{\Fq} =U^{(j-1)}$.
  \end{itemize}

  We will now prove by induction that for every $j\in \{0, \dots, \lceil \frac{k}{rm}\rceil-1\}$, the code $\C^{(j)}$ indeed satisfies the hypotheses~\ref{it:length}--\ref{it:LRC}.  For $j=0$, these properties are clearly satisfied by $\C^{(0)}=\C$.

  Now, fix $j\in \{1, \dots, \lceil \frac{k}{rm}\rceil-1\}$ and assume that $\C^{(j-1)}$ satisfies the hypotheses.
  We have
  \[
  n_j= n_{j-1}-(r_j+1)\geqslant n_{j-1}-(r+1).
  \]
  By Lemma~\ref{lem:prop_punct}, the code $\C^{(j)}= \text{Short}\left(\C^{(j-1)}, W_j', W_j\oplus \vect{u^{(j)}}{\Fq}\right) $ has dimension
  \[
  \begin{aligned}
    k_j
    & =k_{j-1}-\dim_{\Fq}(\pi_{W_j\oplus \vect{ u^{(j)} }{\Fq}}(\C^{(j-1)})) \\
    & =k_{j-1}-\dim_{\Fq}(\pi_{W_j}(\C^{(j-1)}))   \\
    & \geqslant k_{j-1}-r_jm \\
    & \geqslant k_{j-1}-rm.
  \end{aligned}
  \]
  where the last inequality follows from $\dim_{\Fq}(\pi_{W_j}(\C^{(j-1)})) \leq r_jm $.  By induction hypothesis, this proves that~\ref{it:length} and~\ref{it:dim} hold.

  Moreover,
  \[
  n_j-\frac{k_j}{m} = n_{j-1} -  (r_j+1) - \frac{k_{j-1}}{m} + \frac{\dim_{\Fq}(\pi_{W_j}(\C^{(j-1)} ))}{m} \leqslant  n_{j-1} - \frac{k_{j-1}}{m} -1,
  \]
  since $\dim_{\Fq}(\pi_{W_j}(\C^{(j-1)} ))\leqslant r_jm$. By induction hypothesis, this proves~\ref{it:ineq}.

  The code $\C^{(j)}$ being a shortening of $\C^{(j-1)}$, Lemma~\ref{lem:prop_punct} ensures that $d_j\geqslant d_{j-1}$, which proves~\ref{it:dist}. Finally,~\ref{it:LRC} is a direct consequence of Lemma~\ref{lemme_pour_singleton}.

  Now, we apply the Singleton bound to the last nonzero code of the sequence of codes $(\C^{(j)})_j$. By~\ref{it:dim}, we have $k_j \geq 1$ whenever $j\leqslant \frac{k-1}{rm}$. Then, the largest integer $j_{\max}$ satisfying this condition is $j_{\max} = \left\lfloor \frac{k-1}{rm} \right\rfloor = \left\lceil \frac{k}{rm}\right\rceil-1.$ Therefore, we can iterate this procedure at least until we build the code $\C^{(j_{\max})}$.

  We now apply the Singleton bound, as well as~\ref{it:ineq} and~\ref{it:dist}, to $\C^{(j_{\max})}$, in order to get the desired result. If $m \geqslant n_{j_{\max}}$, we directly get
  \[
  d \le d_{j_{\max}}\leqslant n_{j_{\max}}-\frac{k_{j_{\max}}}{m}+1\leqslant n-\frac{k}{m}-j_{\max}+1 =   n-\frac{k}{m}-\left\lceil\frac{k}{rm}\right\rceil+2.
  \]
  Therefore, since $d$ is an integer,
  \[
  d \le \left\lfloor n - \frac{k}{m}-\left\lceil \frac{k}{rm}\right\rceil +2 \right\rfloor =n-\left\lceil\frac{k}{m}\right\rceil+2 - \left \lceil\frac{k}{rm}\right\rceil.
  \]
  Suppose now $m\leqslant n_{j_{\max}}$. The Singleton bound gives $d_{j_{\max}}\leqslant m-\frac{k_{j_{\max}}}{n_{j_{\max}}}+1$. Notice now that $m-\frac{k_{j_{\max}}}{n_{j_{\max}}}\leqslant  n_{j_{\max}}-\frac{k_{j_{\max}}}{m}$ whenever $m \leqslant n_{j_{\max}}$ and $k_{j_{\max}} \le n_{j_{\max}}m$, which is satisfied in our context. Hence, we get the expected bound in the same way as for $m\geqslant n_{j_{\max}}$.
\end{proof}

\begin{remark}
  If $m\geqslant n$ and $r=\frac{k}{m}$, the bound~\eqref{eq:singleton-lrc} becomes the classical Singleton bound, as it is the case in the Hamming metric. Hence, optimal rank LRC codes with locality $r=\frac{k}{m}$ are MRD codes.

  Kadhe \textit{et al.} proved a Singleton-like bound for their definition of rank locality \cite[Theorem 1]{Kadhe_et_al19}. In their case, they assume that $m$ divides $k$, in which case Theorem~\ref{thm:singleton-lrc} gives the exact same bound. Their bound is a direct consequence of the Singleton bound for LRCs in the Hamming metric: a rank-LRC codes $\C$ gives rise to a code $\C_H$ in the Hamming metric with the \emph{same locality}. The relation between our notion of locality and the one in the Hamming metric is not as straightforward (see Lemma~\ref{lem:locality_Hamming_to_rank}, for instance). This is why Theorem~\ref{thm:singleton-lrc} has to be proved from scratch.
\end{remark}

A subfamily of codes given in Example~\ref{ex:famille-localite-1} reaches the Singleton-like bound of Theorem~\ref{thm:singleton-lrc}.

\begin{example}
  Recall that the codes $\C_A$, defined as
  \[
  \C_A = \{ ( M | AM ) \mid M \in \Fq^{m \times n} \} \subseteq \Fq^{m \times 2n}
  \]
  for $n, m \ge 1$ and $A \in \GL_m(\Fq)$, have dimension $mn$ and locality $1$. The Singleton-like bound of Theorem~\ref{thm:singleton-lrc} hence ensures that
  \[
  \dr(\C_A) \le 2n - n - \left\lceil\frac{n}{1}\right\rceil + 2 = 2\,.
  \]
  Some of these codes (e.g., if $A$ is the identity matrix) have minimum rank-distance $1$. Let us exhibit subfamilies with minimum rank-distance $2$. For this sake, fix $\lambda \in \Fq$ and assume that $\lambda$ is not a square in $\Fq$. Fix $A \in \Fq^{m \times m}$ such that $A^2 = \lambda I_m$, where $I_m \in \Fq^{m \times m}$ is the identity matrix.

  Let us now prove that any $C = ( M | AM ) \in \C_A$ has rank $\ge 2$. If $\rk(M) \ge 2$ this is clear, so assume $\rk(M) = 1$. Let $v \in \Fq^m$ be a nonzero column of $M$. Since $A$ is invertible, it suffices to prove that $v$ and $A v$ are not collinear. If they were, the vector $v$ would be an eigenvector of $A$, and the corresponding eigenvalue would be a root of the minimal polynomial $X^2 - \lambda$ of $A$. But $\lambda$ is not a square in $\Fq$, so this leads to a contradiction.
\end{example} 

In the $\Fqm$-linear case, we then obtain the following bound.

\begin{corollary}\label{coro:Singleton-lrc-fqm}
  Let $\C \subseteq  \Fqm^n$ be an $\Fqm$-linear code of dimension $k$, and minimum rank distance $d$. Suppose $\C$ is $r$-rank locally recoverable. Then
  \[
  d\leqslant n-k+2 - \left \lceil\frac{k}{r}\right\rceil.
  \]
\end{corollary}
\begin{proof}
  By Proposition~\ref{prop:vec_to_mat}, by fixing any basis $\Gamma$ of $\Fqm/Fq$, one can associate to $\C$ a code $\Gamma(\C) \subseteq \Hom(\Fq^n, \Fq^m)$ of dimension $K = km$, and minimum distance $d$. This code is also $r$-LRC, so by Proposition~\ref{thm:singleton-lrc}, $d\leqslant n - \frac{K}{m}+2-\left\lceil\frac{K}{rm}\right\rceil = n-k+2-\left\lceil\frac{k}{r}\right\rceil$.
\end{proof}

\begin{remark}
As in the Hamming metric, if $r=k$ we recover the classical Singleton bound, and optimal rank LRC codes with locality $r=k$ are MRD codes.
\end{remark}

\section{Rank-metric analogues of Tamo-Barg optimal LRCs}

In this section, we provide a construction of optimal rank-LRC inspired by Tamo and Barg's construction~\cite{Tamo_2014} in the Hamming metric. Their construction heavily relies on polynomials and Reed--Solomon codes. Likewise, ours will use $q$-polynomials and Gabidulin codes. We recall these notions, describe our construction and provide one instantiation.

\subsection{\texorpdfstring{$q$}{q}-polynomials and Gabidulin codes}

\begin{definition}[Ore polynomials] \label{def:ore}
  Let $m \geq 2$. We denote by $\Fqm\{\tau\}$ the noncommutative ring of \emph{Ore polynomials}
  (also known as \emph{skew polynomials}, or \emph{twisted polynomials}) over $\Fqm$ defined as
  \[\Fqm\{\tau\}\eqdef\left\{ \sum_{i=0}^n c_i \tau^i \,\Big|\, n \geq 0, c_i \in \Fqm\right\},\]
  endowed with the classical additive law, and the multiplication defined by $\tau a= a^q \tau$ for every $a \in \Fqm$.
\end{definition}

In the previous definition, the letter $\tau$ is just a formal variable, without further signification. However, it should be understood as the $q$-Frobenius endomorphism and Ore polynomials then correspond to polynomials in the $q$-Frobenius. This motivates the following terminology.

\begin{definition}
  For $f = \sum_i a_i \tau^i \in \Fqm\set{\tau}$, we define the \emph{$q$-degree} of $f$ as the largest integer $i$ such that $a_i \neq 0$ and we denote it by $\deg_q(f)$.
\end{definition}

By~\cite{Ore33}, the ring $\Fqm\set{\tau}$ is both left and right-Euclidean.

Regarding $\tau$ as an endomorphism of $\Fqm$, we can map Ore polynomials onto endomorphisms of $\Fqm$. More precisely, since $\tau^m=\id_{\Fqm}$, two Ore polynomials induce the same endomorphism of $\Fqm$ if and only if they are congruent modulo $\tau^m-\id$. In particular, every $\Fq$-linear endomorphism of $\Fqm$ is represented by a unique Ore polynomial of $q$-degree less than $m$.

For $f\in \Fqm\set{\tau}$, let $\overline{f}$ denote its unique representative modulo $\tau^{m}-1$ of $q$-degree smaller than $m$.

\begin{proposition}\label{prop:kernel}
  Let $f\in \Fqm\set{\tau}$ be such that $\overline{f}\neq 0$. Then, $\dim_{\Fq} \ker f \leq \deg_q(\overline{f})$.
\end{proposition}

The above proposition is a standard property of linearized polynomials, see for instance~\cite[Theorem~3.50]{Lidl_Niederreiter_1996}.

For every $\Fq$-linear subspace $U$ of $\Fqm$, a family of $q$-polynomials defines a rank-metric code in $\Hom(U,\Fqm)$ as follows.
\begin{definition}[$q$-polynomial code]\label{def:q-pol-codes}
  Let $U$ be an $\Fq$-linear subspace of $\Fqm$ and $\mathcal{F}$ be an $\Fq$-linear subspace of $\Fqm\set{\tau}$. The \emph{$q$-polynomial code associated to $\mathcal{F}$ over $U$} is of the form
  \[
 \pi_U \left(\mathcal{F} \right) = \set{\overline{f}\circ  \iota_U \mid f \in \mathcal{F}} \subseteq \Hom(U,\Fqm).
  \]
\end{definition}
If $\mathcal{F}$ is an $\Fqm$-linear subspace of $\Fqm\set{\tau}$, then $\pi_U(\mathcal{F})$ is naturally an $\Fqm$-linear code. For spaces of $q$-polynomials of sufficiently small $q$-degree, the previous proposition implies that the restriction map $P \mapsto P\circ \iota_U$ is injective and determines a lower bound on the minimum distance of the associated code.

\begin{lemma}\label{lem:parameters_evaluation_q-pol}
  Let $U$ be an $\Fq$-linear subspace of $\Fqm$ of dimension $n$, and $\mathcal{F}$ be an $\Fq$-linear subspace of $\Fqm\set{\tau}$. Assume $\delta\eqdef\max\set{\deg_q(f), f \in \mathcal{F}}< n$. Then, $\pi_U(\mathcal{F})$ has the same $\Fq$-dimension as $\mathcal{F}$, and its minimum rank distance is at least $n-\delta$.
\end{lemma}

\begin{proof}
  Since $\delta < n \leq m$, the map $f \mapsto \pi_U(\overline{f})$ is injective over $\mathcal{F}$ by Proposition~\ref{prop:kernel}. This proves the equality of dimensions. The lower bound on the minimum rank distance comes from the rank-nullity theorem and Proposition~\ref{prop:kernel}.  
\end{proof}

The most typical examples of $q$-polynomial codes are Gabidulin codes, originally introduced by Delsarte~\cite{Delsarte} and Gabidulin~\cite{Gabidulin}.
For a positive integer $k$, we denote by $\F_{q^m}\{\tau\}_{<k}$ the vector space of $q$-polynomials of $q$-degree less than $k$.

\begin{definition}
  For an $\Fq$-linear subspace $U$ of $\Fqm$ of dimension $n\leq m$, the associated \emph{Gabidulin code} of dimension $k \le n$ is
  \[
  \Gab_k(U) \eqdef \pi_U\big(\F_{q^m}\{\tau\}_{<k}\big).
  \]
\end{definition}

Gabidulin codes are rank-metric analogues of Reed--Solomon codes. By Lemma~\ref{lem:parameters_evaluation_q-pol}, they are MRD codes, so they share with Reed--Solomon codes the remarkable property of being optimal with respect to the Singleton bound. Besides, for every $\Fq$-linear subspace $U$ of $\Fqm$ of dimension $n\leq m$,  the code $\Gab_k(U)$ (with $k < n$) is $k$-rank locally recoverable thanks to Corollary~\ref{coro:localite_MRD}. In particular, for every nonzero $u \in U$, every $k$-dimensional subspace $S$ of $U$ such that $u\notin S$ is a helper space for $u$.

\subsection{Construction}

We build the rank-metric analogue of the general construction of optimal $r$-LRCs  in the Hamming metric given by Tamo and Barg in~\cite[\textsection A]{Tamo_2014}. This seminal construction consists in designing an evaluation domain $\mathcal{A} = A_1 \cup \dots \cup A_t \subseteq \Fq$, with pairwise disjoint subsets $A_i$ of size $r+1$, along with polynomials $f \in \Fq[x]$ that evaluate like polynomials of degree less than $r$ on each of the $A_i$. For any of these polynomials $f$, the vector of evaluations of $f$ over $A_i$ lies in a Reed--Solomon code of dimension $r$, therefore each of its coordinates can be recovered by polynomial interpolation.

Tamo and Barg's trick to design such polynomials is to rely on a so-called \emph{good polynomial for the partition $(A_i)_{i \in [t]}$}: a polynomial $g$ of degree $r+1$, that is constant on each of the $A_i$. Then polynomials $f$ of the form $f=\sum_{j\geq 0} f_j g^j$ with $\deg(f_j) < r$ have the desired property.

We mimic this idea in the context of $q$-polynomials. Now, the evaluation domain is an $\Fq$-vector space $U$. Given a \enquote{partition} of $U$ in $(r+1)$-dimensional subspaces $U_i$ (more precisely an $r$-spread of the projectivized space $\mathrm{PG}(U)$), we aim to design a family of $q$-polynomials that evaluate as $q$-polynomials of degree less than $r$ on each $U_i$. This way, we are able to recover their values at some $u \in U_i$ thanks to any subspace $S \subseteq U_i$ such that $U_i= S \oplus \vect{ u }{\Fq}$. To do so, we also rely on a good $q$-polynomial as described in the next theorem.

\begin{theorem}\label{thm:TB}%
  Let $U \subseteq \Fqm$ of dimension $n$ and $V=\Fqm$. Let $r$ be a positive integer such that $r+1$ divides $n$. Assume there exist
  \begin{enumerate}[label=(\alph*)]
  \item\label{it:partition} a family $\mathcal{U}$ of $t= \frac{q^n-1}{q^{r+1}-1}$ $\Fq$-subspaces $U_1,\:\dots,\:U_t$ of $U$ of dimension $r+1$ such that $U =\bigcup_{i=1}^t U_i$ (so $U_{i} \cap U_j=\set{0}$ for every $i\neq j$),
  \item\label{it:good_poly} a $q$-polynomial $g$ of degree $r+1$ such that for every $i \in [t]$, there exists $\lambda_i \in \Fqm$ such that $g_{|U_i} = \lambda_i \id_{U_i}$.
  \end{enumerate}
  Let $k<n$ such that $r$ divides $k$ and $k+\frac{k}{r}-2 < n$. Define
  \[
  \mathcal{F} \eqdef \set{\sum_{j=0}^{\frac{k}{r}-1} f_j \circ g^j  \,\big|\, f_j \in \Fqm\set{\tau}_{<r}}\,.
  \]

  Then, the Tamo--Barg-like code
  \[
  \TB(\mathcal{U}, g, k) \eqdef \pi_U(\mathcal{F})\subseteq \Hom(U,\Fqm)
  \]
  is an optimal rank $r$-locally recoverable code of dimension $k$ over $\Fqm$.
\end{theorem}
\begin{proof}
  It is clear that $\dim_{\Fqm} \mathcal{F} \leq \frac{k}{r} \cdot r = k$. Now, by unicity of the right-division in $\Fqm\set{\tau}$ by $g$, we get $\dim_{\Fqm} \mathcal{F}=k$. Moreover, all the $q$-polynomials in $\mathcal{F}$ have $q$-degree bounded from above by
  \[
  \delta\eqdef r-1+\left(\frac{k}{r}-1\right)(r+1)=k+\frac{k}{r}-2<n.
  \]
  Therefore, by Lemma~\ref{lem:parameters_evaluation_q-pol}, the code $\TB(\mathcal{U}, g, k)$ has dimension $k$ and its minimum distance satisfies $\dr(\TB(\mathcal{U}, g, k)) \geq n-\delta=n-k-\frac{k}{r}+2$, which matches the Singleton bound for $\Fqm$-linear LRCs (Corollary~\ref{coro:Singleton-lrc-fqm}).

  It thus remains to prove that $\TB(\mathcal{U}, g, k)$ is indeed an $r$-rank LRC. Take $u \in \Fqm^*$. The hypothesis~\ref{it:partition} ensures there exists a unique $i \in [t]$ such that $u \in U_i$. By~\ref{it:good_poly}, there exists $\lambda_i \in \Fqm$ such that for every $f \in \TB(\mathcal{U}, g, k)$, we have
  \[
  f_{|U_i}= \sum_{j=0}^{\frac{k}{r}-1} f_j \circ g^j \circ \iota_{U_i} = \sum_{j=0}^{\frac{k}{r}-1} f_j\circ(\lambda_i^j \id_{U_i}).
  \]

  Therefore $\pi_{U_i}(\TB(\mathcal{U}, g, k)) \subseteq \Gab_r(U_i)$. Moreover, as the $q$-polynomials $f_j$ run in the whole space of polynomials of $q$-degree less than $r$, we have $\pi_{U_i}(\TB(\mathcal{U},g,k)) = \Gab_r(U_i)$. The latter code being MRD, Corollary~\ref{coro:info_space_MRD} ensures that any subspace $S \subseteq U_i$ of dimension $r$ such that $U_i = S \oplus \vect{ u }{\Fq}$ is an information space for $\Gab_r(U_i)$, i.e.,  $\dim_{\Fqm}(\pi_S(\Gab_r(U_i))) = \dim_{\Fqm}(\Gab_r(U_i))$. Then
  \[
  \pi_{S}(\TB(\mathcal{U},g,k))=\pi_S\left(\pi_{U_i}(\TB(\mathcal{U},g,k))\right) = \pi_S(\Gab_r(U_i))
  \]
  has the same dimension as $\pi_{U_i}(\TB(\mathcal{U},g,k))$, which means that $S$ is a helper space for $u$ in $\TB(\mathcal{U},g,k)$.

\end{proof}

\begin{remark}
  Note that for every $U_i \in \mathcal{U}$, we have $U_i = \ker(g-\lambda_i \id)$, since the right-hand side has dimension at most $r+1$, and $\dim_{\Fq}(U_i)=r+1$. Then $\lambda_i \neq \lambda_j$ for $i\neq j$.
\end{remark}

\begin{remark}
  The set of \emph{good} $q$-polynomials (that satisfy~\ref{it:good_poly}) is an $\Fqm$-vector space. Actually, good $q$-polynomials are tightly related to good polynomials for the original Tamo-Barg construction~\cite{Tamo_2014}. Indeed, seen as classical polynomials via the identification $\tau=X^q$, $q$-polynomials are always divisible by $X$. Then, if $g$ satisfies~\ref{it:good_poly}, the polynomial $g(X)/X$ (of degree $q^{r+1}-1$) is a good polynomial in the sense of~\cite[\textsection A]{Tamo_2014} for the partition $\mathcal{U}^*\eqdef \set{U_i\setminus \set{0}, U_i \in \mathcal{U}}$.
\end{remark}

We instantiate this construction using a Desarguesian spread, leaving the existence of alternative constructions as an open question.

\begin{example}\label{ex:desarguian}
  Take $U=\Fqm$ and $\alpha_1,\dots,\alpha_t$ some representatives of the classes of $\Fqm^*/\F_{q^{r+1}}^*$. Then, the family $\mathcal{U}=\set{U_i\eqdef \alpha_i \F_{q^{r+1}} \mid i \in [t]}$ satisfies the condition~\ref{it:partition} of Theorem~\ref{thm:TB}.

  Fix $\eta, \nu \in \Fqm$. On each $U_i$, the $q$-polynomial $g=\eta\tau^{r+1}+\nu \id$ satisfies the condition~\ref{it:good_poly} of Theorem~\ref{thm:TB}: for every $i \in [t]$ and $x \in \F_{q^{r+1}}$, we have
  \[
  g(\alpha_i x) = \eta \alpha_i^{q^{r+1}} x+\nu \alpha_i x
  = \big(\eta\alpha_i^{q^{r+1}-1}+\nu\big) \alpha_i x
  =\lambda_i \alpha_i x.
  \]
  
  Note that a $q$-polynomial $g$ satisfying the condition~\ref{it:good_poly} with respect to $\mathcal{U}$ implies that it is L-$q^{r+1}$-partially scattered (of index $0$) in the sense of~\cite{Bartoli2022investigating}, that is for any $y, z \in \Fqm^*$,
  \[
  \frac{g(y)}{y}=\frac{g(z)}{z} \implies \frac{y}{z}\in \F_{q^{r+1}}.
  \]

  By~\cite[Proposition~3.1]{Bartoli2022investigating}, since $r+1$ divides $m$, a monomial of the form $\tau^{t}$ is L-$q^{r+1}$-partially scattered of index $0$ if and only if $t \mid r+1$. This indicates that the only good $q$-polynomials with respect to $\mathcal{U}$ are the ones considered above.
  
  Note that we recover~\cite[Construction~1]{Kadhe_et_al19} with $\eta=1$ and $\nu=0$. As discussed in Section~\ref{subsec:comparison}, the notion of local recoverability studied in~\cite{Kadhe_et_al19} is not the same as ours. Nevertheless, this specific construction gives an optimal LRC for both notions.
\end{example}

\begin{remark}
  The construction of Theorem~\ref{thm:TB} shares some analogies with the one of~\cite[\textsection 3.1]{loc2}, which provides optimal LRCs in the sense of~\cite{Kadhe_et_al19}. The main difference lies in the fact that the partition of the condition~\ref{it:partition} is replaced by a decomposition as a direct sum
  \begin{equation}\label{eq:direct_sum}%
    U=U_1\oplus \dots \oplus U_s
  \end{equation}
  of $(r+1)$-dimensional subspaces and that a \emph{good $q$-polynomial} for such a direct sum satisfies that $g_{|U_i}=\lambda_i \id_{U_i}$ for an homothety factor $\lambda_i\in \Fq$ (instead of $\Fqm$ as in~\ref{it:good_poly}). In other words, the endomorphism associated to $g$ is diagonalizable on $U$. In this case, with a basis of $U$ formed by the union of bases of the eigenspaces $U_i$, the value of any $q$-polynomial of $\mathcal{F}$ at a basis vector can be recovered by its restriction on a given $U_i$, which lies in the Gabidulin code $\Gab_{r}(U_i)$. In~\cite[\textsection 3.1]{loc2}, the good polynomial is the $q$-polynomial $\phi_T$ associated to a Drinfeld module $\phi$ of rank $r+1$.
  
  Note that these codes are not rank $r$-locally recoverable for our definition. The relaxed condition~\eqref{eq:direct_sum} does not allow the recoverability of any nonzero of $U$ with a helper space of dimension $r$, but only of the ones which lie in one of the $U_i$'s.
\end{remark}

\section*{Acknowledgements}

The authors warmly thank Gianira Alfarano for her inputs regarding the use of the Desarguian spread in Example~\ref{ex:desarguian}. This work was supported by a \emph{Research AAP} from University Paris 8, allowing the authors to meet there.

CG is financially supported by the Military French Ministry – Defense and Innovation Agency (DGA-AID). JN is supported by the French government \textit{Investissements d’Avenir} program ANR-11-LABX-0020-01.

\bibliographystyle{plain}
\bibliography{biblio}

\end{document}